\documentclass[a4paper]{article}
\usepackage[margin=3.3cm]{geometry}
\usepackage{floatrow}
\usepackage{graphicx}
\usepackage{picinpar}
\usepackage{oubraces}
\usepackage{subcaption}
\usepackage[nodisplayskipstretch]{setspace} %\setstretch{1.5}

\usepackage{tikz}
\usetikzlibrary{calc}

\newcommand{\tikzmark}[1]{\tikz[overlay,remember picture] \node (#1) {};}
\newcommand{\DrawBox}[4][]{%
    \tikz[overlay,remember picture]{%
        \coordinate (TopLeft)     at ($(#2)+(-0.2em,0.9em)$);
        \coordinate (BottomRight) at ($(#3)+(0.2em,-0.3em)$);
        \path (TopLeft); \pgfgetlastxy{\XCoord}{\IgnoreCoord};
        \path (BottomRight); \pgfgetlastxy{\IgnoreCoord}{\YCoord};
        \coordinate (LabelPoint) at ($(\XCoord,\YCoord)!0.5!(TopLeft)$);
        \draw [red,#1] (TopLeft) rectangle (BottomRight);
        \node [left, #1, fill=none, fill opacity=1] at (LabelPoint) {#4};
    }
}
\newcommand{\DrawLowBox}[4][]{%
    \tikz[overlay,remember picture]{%
        \coordinate (TopLeft)     at ($(#2)+(-0.2em,0.9em)$);
        \coordinate (BottomRight) at ($(#3)+(0.2em,-0.3em)$);
        \path (TopLeft); \pgfgetlastxy{\XCoord}{\IgnoreCoord};
        \path (BottomRight); \pgfgetlastxy{\IgnoreCoord}{\YCoord};
        \coordinate (LabelPoint) at ($(\XCoord,\YCoord)!0.25!(TopLeft)$);
        \draw [red,#1] (TopLeft) rectangle (BottomRight);
        \node [left, #1, fill=none, fill opacity=1] at (LabelPoint) {#4};
    }
}

\usepackage{amsmath,amssymb,amsthm}
\usepackage{multirow}
\usepackage{algorithm2e}
\newtheorem{lemma}{Lemma}
\newtheorem{theorem}{Theorem}
\newtheorem{corollary}{Corollary}
\newtheorem{remark}{Remark}

\newcounter{THM:2-connected}
\newcounter{THM:arbitrary}
\title{C-Planarity of Overlapping Clusterings Including Unions of Two Partitions\footnote{This research has received funding from the European Research Council under the European Union’s Seventh Framework Programme (FP7/2007-2013) / ERC grant agreement no. 319209.}}
\author{Jan C. Athenst\"adt and Sabine Cornelsen \\ University of Konstanz}

\begin{document}

\maketitle

\begin{abstract}
  We show that clustered planarity with overlapping clusters as
  introduced by Didimo et al.~\cite{didimo_etal:2008} can be solved in
  polynomial time if each cluster induces a connected subgraph. It can
  be solved in linear time if the set of clusters is the union of two
  partitions of the vertex set such that, for each cluster, both the
  cluster and its complement, induce connected subgraphs. Clustered
  planarity with overlapping clusters is NP-complete, even if
  restricted to instances where the underlying graph is 2-connected,
  the set of clusters is the union of two partitions and each cluster
  contains at most two connected components while their complements
  contain at most three connected
  components~\cite{athenstaedt_etal:gd2014}.
\end{abstract}

\section{Introduction}\label{SEC:introduction} % and background}

An \emph{(overlapping) clustered graph} $(G=(V,E),\mathcal C)$
consists of an undirected graph $G$ and a set $\mathcal C$ of
\emph{clusters}, i.e., of subsets of the vertex set $V$. A vertex may
be contained in several clusters. Moreover, clusters may overlap,
i.e., there might be $C_1,C_2 \in \mathcal C$ with
$C_1 \cap C_2 \neq \emptyset$, $C_1 \not\subseteq C_2$, and
$C_2 \not\subseteq C_1$.  Didimo et al.~\cite{didimo_etal:2008}
defined planarity for overlapping clustered graphs geometrically: An
overlapping clustered graph $(G=(V,E),\mathcal C)$ is clustered planar
if the vertices can be represented by distinct points, each edge
$e \in E$ by a curve $R(e)$, and each
%or
%
cluster $C \in \mathcal C$ by a simple closed region
$R(C)$ in the plane such that for $X,Y \in E \cup \mathcal C$ we have that
%\begin{itemize}
%\item 
(i) $X \subset R(X)$, $(V\setminus X) \cap R(X) = \emptyset$, 
%for $X \in E \cup \mathcal C$
%\item 
%for any two $X,Y \in E \cup \mathcal C$ we have that
  %\begin{itemize}
  %\item 
(ii) $R(X) \subseteq R(Y)$ if $X \subseteq Y$, and
  %\item 
(iii) every connected region of $R(X) \cap R(Y)$ contains a vertex.
  %\end{itemize}
%\end{itemize}
%
%
%
E.g., the clustered graph in Fig.~\ref{FIG:ex1} is clustered
planar while the clustered graph in Fig.~\ref{FIG:ex2} is not.

Clustered planarity is NP-complete in general as shown
in~\cite{johnson/pollak:87}, where the case with $E = \emptyset$ is
examined. In \cite{didimo_etal:2008}, it was posed as an open question
whether clustered planarity is polynomial-time solvable for
overlapping clustered graphs if each cluster induces a connected
subgraph. We will answer this question in the affirmative.

% Clustered planarity with overlapping clusters should not be confused
% with the concept of region-region crossings as introduced by Angelini
% et al.~\cite{angelini_etal:cg2015}. In fact, the drawing of the
% clustered graph in Fig.~\ref{FIG:ex2} is planar, has no edge-region
% crossings (i.e., the clustered planarity conditions on the regions are
% fulfilled if $C$ is an edge) and only one region-region crossing.
% However, the intersection of the two regions does not contain a vertex.
% %as required by the definition of clustered planarity.

%\paragraph{Related Work:} 

\begin{figure}
  \begin{minipage}[b]{.45\linewidth}
  \begin{center}
     \includegraphics[page=2]{intro}%
   \end{center}
     \subcaption{clustered planar}\label{FIG:ex1}
  \end{minipage}
  \begin{minipage}[b]{.45\linewidth}
     \centering{\includegraphics[page=1]{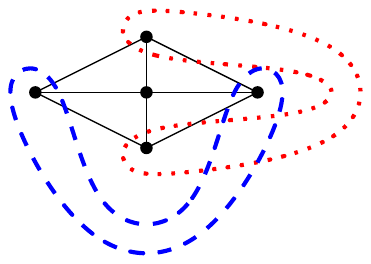}}%
     \subcaption{not clustered planar}\label{FIG:ex2}
  \end{minipage}
\caption{Two graphs with two clusters each (set of vertices enclosed by red dotted and blue dashed curve, respectively)}
\end{figure}
If the clustering is hierarchical, i.e., if any two clusters in
$\mathcal C$ are either disjoint or one is contained in the other then
clustered planarity is the classical problem of c-planarity as
considered in~\cite{feng_etal:esa95}.  One of the most important open
problems in the field of Graph Drawing is the complexity of
c-planarity of hierarchically clustered graphs. An overview on the
classical c-planarity problem can be found
in~\cite{cortese_etal:2008,patrignani:gdhandbook,blaesius/rutter:2016}.
Dahlhaus~\cite{dahlhaus:latin98} and later Cortese et
al.~\cite{cortese_etal:2008} showed, that c-planarity of
hierarchically clustered graphs can be solved in linear time if each
cluster induces a connected subgraph. Their approaches make use of the
decomposition of the graph into 3-connected components as represented
by BC- and SPQR-trees%
%, but rely heavily on the hierarchical structure of the clusters
. 
%They are thus not directly applicable when clusters may overlap.

Angelini et al.~\cite{angelini_etal:cg2015} defined drawings with
region-region crossings of hierarchically clustered graphs. These are
essentially representations by points and regions such that all
conditions of clustered planarity are fulfilled except for Condition
(iii) when $X$ and $Y$ are both clusters. E.g., Fig.~\ref{FIG:ex2}
shows a drawing of a clustered graph with one region-region crossing.
Observe that the intersection of the two regions does not contain a
vertex as required by the definition of clustered planarity.  Angelini
et al.~\cite{angelini_etal:cg2015} showed how to use SPQR-trees to
test in polynomial time whether any hierarchically clustered graph
with an underlying 2-connected graph has a drawing with region-region
crossings.

If $E = \emptyset$ then clustered planarity is closely
related to the NP-complete problem of hypergraph (vertex) planarity as defined
in~\cite{johnson/pollak:87}: Given a set $\mathcal C$ of subsets of a
set $V$, is there a planar support, i.e., a planar graph $G=(V,E)$
such that each set in $\mathcal C$ induces a connected subgraph of
$G$. Various subclasses of planar supports that directly imply
clustered planarity~--~such as trees, cacti, and outerplanar supports~--~were
considered~\cite{bixbi/wagner:1988,vanBeveren_etal:gd2016,brandes_etal:iwoca10outer,brandes_etal:2011path,buchin_etal:2011,kaufmann_etal:gd2008}.
Hypergraph planarity remains NP-complete even if $\mathcal C$ is the
union of two partitions~\cite{athenstaedt_etal:gd2014}.
%
%\paragraph{Contribution:}
%
%We show in Sect.~\ref{SEC:np-complete} that the proof in
%
%\textbf{Contribution:} 
%
The proof in \cite{athenstaedt_etal:gd2014} even shows that clustered
planarity remains NP-complete if the underlying graph $G$ is
2-connected, $\mathcal C$ is the union of two partitions, each cluster
contains at most two connected components, and the complement of any
cluster 
%contains 
at most three connected components.

\textbf{Contribution:} 
In Sect.~\ref{SEC:co-connected} we focus on the
union of two partitions. We further require that for each cluster,
both the cluster itself and its complement are connected. Different
from hierarchical clusterings, this connectivity property does not
automatically imply clustered planarity in the overlapping case. Yet,
for the union of two partitions, we can give a characterization that
yields a linear-time testing algorithm. Finally, in
Sect.~\ref{SEC:c-connected-SPQR} and \ref{SEC:c-connected-BC}, we show
how to use BC-trees, SPQR-trees and the consecutive-ones property to
obtain an algorithm for testing clustered planarity of possibly
overlapping but connected clusters. The run time of the algorithm is
polynomial in $|V|$ and $|\mathcal C|$.
 
\section{Preliminaries}\label{SEC:preliminaries}

For a subset $C \subseteq V$ of the vertices of an undirected graph
$G=(V,E)$, we denote by $G[C]$ the subgraph of $G$ induced by $C$,
i.e.\ the graph with vertex set $C$ and edge set $\{e \in E;\; e
\subseteq C\}$. A $C$-path ($C$-cycle) is a path (simple
cycle) in $G[C]$.  A \emph{partition} of $V$ is a set
$\mathcal P$ of subsets of $V$ such that each vertex in $V$ is
contained in exactly one set in $\mathcal P$. For two partitions
$\mathcal P_B = \{B_1, \ldots , B_{\ell_B}\}$ and $\mathcal P_R =
\{R_1, \ldots , R_{\ell_R}\}$ of $V$, we define the \emph{intersection
  partition} $\mathcal P_I = \{B_i \cap R_j; i = 1, \ldots , \ell_B, j
= 1, \ldots , \ell_R\}$.  The \emph{connected intersection partition}
of $\mathcal P_B$ and $\mathcal P_R$ is the partition induced by the
connected components of $G[C], C \in \mathcal P_I$. 

A \emph{consecutive-ones ordering} of a binary matrix is a permutation
of its columns such that in each row all of the 1s are consecutive,
i.e. such that each row is of the form $0^*1^*0^*$.  A binary matrix
has the \emph{consecutive-ones property} if and only if it has a
consecutive-ones ordering. It can be tested in linear time whether a
binary matrix has the consecutive-ones property and a consecutive-ones
ordering can be found in linear time if it
exists~\cite{booth/lueker:1976}.

\subsection{Planarity of Overlapping Clustered Graphs}

Let $(G=(V,E),\mathcal C)$ be
an overlapping clustered graph. Let $\ell(\mathcal C) = \sum_{C \in \mathcal
  C}|C|$ be the total size of all clusters. The clustered graph $(G,\mathcal{C})$ is
\emph{c-connected} if $G[C]$ is connected for all $C \in
\mathcal{C}$ and \emph{c-co-connected} if both, $G[C]$
and $G[V \setminus C]$, are connected for all $C \in \mathcal{C}$.

If $(G,\mathcal{C})$ is c-connected then a \emph{c-planar embedding} of
$G$ for $\mathcal C$ is a planar embedding of $G$ such that $V
\setminus C$ is in the outer face of $G[C]$ for all $C \in
\mathcal{C}$.
A graph $G^+=(V,E^+)$ is a \emph{c-planar support} of a clustered
graph $(G=(V,E),\mathcal C)$ if $E \subseteq E^+$, $(G^+,\mathcal{C})$
is c-connected and there is a c-planar embedding of $G^+$ for
$\mathcal C$. A clustered graph is \emph{c-planar} if and only if it
has a c-planar support. 

It was shown that a c-connected clustered
graph~\cite{didimo_etal:2008} or a hierarchically clustered
graph~\cite{feng_etal:esa95}, respectively, is clustered
planar in the sense of \cite{didimo_etal:2008} if and only if it has a
c-planar support.

% Note that Didimo et al.~\cite{didimo_etal:2008} actually called c-planar
% overlapping clustered graphs oc-planar.)

\subsection{BC-Trees}

A vertex $v$ is a \emph{cut vertex} of a connected graph $G$ if the
graph that results from $G$ by deleting $v$ and its incident edges is
not connected. A connected graph is \emph{2-connected} if it contains more than
two vertices but no cut vertices. The \emph{blocks} of a connected graph
are the maximal 2-connected subgraphs and the subgraphs induced by
bridges. The vertices of the \emph{block--cut tree (BC-tree)} of
a graph $G$ are the blocks and the cut vertices of $G$. There is an edge
in the block--cut tree between a block $H$ and a cut vertex $v$
if $v$ is contained in $H$.

\subsection{SPQR-Trees}

Two vertices $v$ and $w$ are a \emph{separation pair} of a 2-connected
graph $G$ if the graph that results from $G$ by deleting $v$ and $w$
and their incident edges is not connected.  A graph is
\emph{3-connected} if it contains more than three vertices but no
separation pair.  An
\emph{SPQR-tree}~\cite{dibattista/tamassia:icalp90} is a labeled tree
that represents the decomposition of a 2-connected graph into
3-connected components.  Each node $\nu$ of an SPQR-tree is labeled with
a multi-graph skel$(\nu)$ -- called the \emph{skeleton} of $\nu$. There are
four different types of labels with the skeletons:
\emph{S-nodes} for simple cycles, 
\emph{P-nodes} for three or more parallel edges,
\emph{R-nodes} for a simple 3-connected graph, and
\emph{Q-nodes} for two parallel edges.

The Q-nodes are the leaves of an SPQR-tree. No two S-nodes, nor two
P-nodes are adjacent in an SPQR-tree. For each node $\nu$ of an
SPQR-tree there is a one-to-one correspondence of the edges of
skel$(\nu)$ and the edges incident to $\nu$ (except for the Q-nodes
where one of the two edges of the skeleton corresponds to the only
incident edge of the Q-node). The edge of skel$(\nu)$ corresponding to
the edge $\{\nu,\mu\}$ of the SPQR-tree is denoted by $e_\mu$.  We
consider the edges of the skeletons oriented. For simplicity, we
assume that the edges of the skeleton of an S-node are oriented as a
directed cycle and the edges of the skeleton of a P-node are all
oriented in parallel.

% We
% consider the vertices of the skeletons labeled such that the vertices
% within one skeleton have distinct labels and such that for any edge
% $\{\nu,\mu\}$ of the SPQR-tree, the end vertices of the edge $e_\mu$
% in skel$(\nu)$ and 
% %of the end vertices of the edge 
% $e_\nu$ in
% skel$(\mu)$ have the same labels.

We consider the SPQR-tree $T$ rooted at a Q-node $r$. Let $\nu$ be a
node of $T$. The \emph{root edge} of skel$(\nu)$ is the edge that
corresponds to the parent edge of $\nu$. The \emph{poles} of skel$(\nu)$ 
(or node $\nu$, respectively)
are the end vertices of the root edge.   Let skel$^-(\nu)$ be the
skeleton of $\nu$ without the root edge. 
%Further, let $\{\nu,\mu\}$ be an edge of an SPQR-tree
%and let $e_\nu$ and $e_\mu$ be the edges of $G_\nu$ and $G_\mu$,
%respectively, that are assigned to $\{\nu,\mu\}$. Then, $e_\nu$ and $e_\mu$
%have the same end-vertices -- say $s$ and $t$. 
%Moreover, let $x'$
%and $y'$ be two nodes in different connected components of $T$
%without the edge $\{x,y\}$. Then $G_{x'}$ and $G_{y'}$ share at most
%$u$ and $v$ as common vertices.
%
Each node $\nu$ of the rooted SPQR-tree represents a (multi-)graph
$G_r(\nu)$: The Q-nodes (excluding the root) represent a graph with
two vertices connected by an edge and additionally by the root
edge. Let $\nu$ be a non-leaf node of an SPQR-tree and let
$\nu_1,\dots,\nu_k$ be the children of $\nu$. For $i=1,\dots,k$,
remove the edge associated with $\{\nu,\nu_i\}$ from both skel$(\nu)$
and $G_r(\nu_i)$. Insert the remaining parts of $G_r(\nu_i)$ into
skel$(\nu)$ identifying the poles of $G_r(\nu_i)$ with its counter
parts in skel$(\nu)$.
The \emph{poles} of $G_r(\nu)$ are the poles of $\nu$. 
Let $G^-(\nu)$ be $G(\nu)$
without the root edge of skel$(\nu)$.
%The graph $G_r(r)$ is the graph represented by the whole SPQR-tree. 
The edges of $G_r(r)$ correspond to the Q-nodes of the SPQR-tree.

Every 2-connected graph is represented by a unique SPQR-tree (up to
the choice of the root) and the SPQR-tree of a 2-connected graph can
be constructed in linear time~\cite{gutwenger/mutzel:gd2000}.

\section{Two C-Co-Connected Partitions}\label{SEC:co-connected} 

In this section we show that c-planarity of a c-co-connected clustered
graph can be tested in linear time if the set of clusters is the union
of two partitions. Observe that in contrast to the hierarchical
case~\cite{cornelsen/wagner:wg2003}, there are c-co-connected clustered
graphs with an underlying planar graph that are not c-planar. 
  E.g., the graph $G=(V,E)$ in Fig.~\ref{FIG:co-connected_not_c-planar} is
  3-connected and, thus, has a unique embedding up to the choice of
  the outer face. No matter which face we choose as the outer face,
  there is always at least one cluster $C$ among the four clusters in
  $\mathcal P_B \cup \mathcal P_R$ such that $G[C]$ contains a simple
  cycle enclosing a vertex in $V \setminus C$.
\begin{figure}[t!]
     \centering{\includegraphics[page=7]{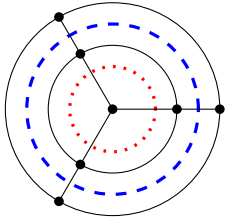}}%
\caption{$\mathcal P_B$ contains the  two clusters separated by the blue
  (dashed) curve, $\mathcal P_R$ contains the two clusters
  separated by the red (dotted) curve. $(G, \mathcal
  P_B \cup \mathcal P_R)$ is c-co-connected and $G$ is planar but $(G,
  \mathcal P_B \cup \mathcal P_R)$ is not c-planar.}
\label{FIG:co-connected_not_c-planar}
\end{figure}

The key for the algorithm is the following characterization.

\begin{theorem}\label{THEO:2partitions}
  Let $G=(V,E)$ be a graph and let $\mathcal P_R$ and $\mathcal P_B$
  be two partitions of $V$ such that the clustered graph $(G,\mathcal
  P_R \cup \mathcal P_B)$ is c-co-connected. Let $\mathcal P'_I$ be the
  connected intersection partition of $\mathcal P_R$ and $\mathcal
  P_B$. Then $(G,\mathcal P_R \cup \mathcal P_B)$ is c-planar if and
  only if $(G,\mathcal P_I')$ is c-planar.
\end{theorem}
\begin{proof}
We first show that \emph{if $(G, \mathcal P_B \cup \mathcal P_R)$ is c-connected and c-planar then $(G, \mathcal P'_I)$
  is c-planar}: Consider a c-planar embedding of $G$ for $\mathcal P_R \cup \mathcal
  P_B$.  Let $C \in \mathcal P'_I$, $R \in \mathcal P_R$, and $B
  \in \mathcal P_B$ with $C \subseteq B \cap R$. Assume there is a
  vertex $v$ in an inner face of $G[C]$. Observe that by c-planarity
  of $(G,\mathcal P_R\cup \mathcal P_B)$, all vertices in the inner
  faces of $G[C]$ are in $B \cap R$. Since $G[B]$ and $G[R]$ are
  connected there must be a path from $v$ to $C$ that contains only
  vertices in one inner face of $G[C]$ and thus in $B \cap R$. Therefore,
  $v$ and $C$ are in the same connected component of $G[B \cap
  R]$ and hence $v \in C$.

  \begin{window}[6,r,\includegraphics[page=6]{negative_examples},]
    We now show that \emph{if $(G, \mathcal P_B \cup \mathcal P_R)$ is
      c-co-connected and $(G, \mathcal P'_I)$ is c-planar then $(G,
      \mathcal P_B \cup \mathcal P_R)$ is c-planar}: Let $C \in
    \mathcal P_B \cup \mathcal P_R$ (squared blue vertices in the
    drawing). c-co-connected implies that $V \setminus C$ is contained
    in one face $f$ of $G[C]$ (shaded area in the drawing). We call
    $C$ \emph{bad} if $V \setminus C$ is contained in an inner face of
    $G[C]$. Among all planar embeddings of $G$ that are c-planar for
    $\mathcal P_I'$ choose one that minimizes the number of bad
    clusters. Assume there is a bad $C \in \mathcal P_B \cup \mathcal
    P_R$. We assume without loss of generality that $C \in \mathcal
    P_B$.  We show that this would yield a contradiction to the choice
    of the embedding.

    \indent
    If $C$ is the union of some clusters in $\mathcal P_R$, choose a
    face $f_0$ of $G$ inside $f$ incident to a vertex of $C$ as the
    outer face, decreasing the number of bad clusters.  
    
    \indent
    Otherwise, let $C' \in \mathcal P_R$ intersect $C$ and $V
    \setminus C$.
    Since $G[C']$ is connected, $E(C' \cap C, C' \setminus C)$ is not
    empty.  There must even be an edge $e \in E(C' \cap C, C'
    \setminus C)$ that is in the outer face of $G[C']$: Otherwise $G[C
    \cap C']$ would enclose $C' \setminus C$, i.e., there is a cycle
    $c$ in $G[C \cap C']$ with a vertex in $C' \setminus C$ in its
    inside.  Thus $f$ is contained in the region bounded by
    $c$. However, $c$ is contained in a connected component of $C \cap
    C'$. This contradicts the fact that $(G,\mathcal P_I')$ is
    c-planar.
    Let now $f_0$ be a face of $G$ incident to $e$ in the outer face
    of $G[C']$.
    
    \indent Now $f_0$ is incident to a vertex of the outer face of
    both a graph induced by a cluster in $\mathcal P_R$ and a graph
    induced by cluster in $\mathcal P_B$.  Thus, $f_0$ is not
    contained in any cluster.  Choosing $f_0$ as the outer face
    decreases the number of bad clusters. This contradicts that we
    have chosen a planar embedding minimizing the number of bad
    clusters. \qedhere
    \end{window}
\end{proof}

Since c-planarity for c-connected hierarchically clustered graphs can
be tested in linear time~\cite{cortese_etal:2008}, it remains to show
that $\mathcal P_I'$ can be constructed in linear time. Since
connected components can be computed in linear time it suffices to
show that the intersection partition $\mathcal P_I$ of two partitions
$\mathcal P_B = \{B_1, \ldots , B_{\ell_B}\}$ and $\mathcal P_R =
\{R_1, \ldots , R_{\ell_R}\}$ of $V$ can be computed in linear time.
We introduce the following data structure: For $X \in \{B,R\}$, we use
a vertex array with $X[v]=i$ for $v \in X_i$. We also initialize an
array $S[1,\dots,\lambda_R]$ of stacks, where $S[i]$ will contain the
vertices of $R_i$, $i=1,\dots,\lambda_R$ in the order in which they
appear in $B_1,\dots,B_{\lambda_B}$. We fill the stacks as follows:
For $i=1,\dots,\lambda_B$ and $v \in B_i$, we push $v$ to $S[R[v]]$.
Now, the sets in $\mathcal P_I$ can be obtained by going through the
stacks and opening a new set whenever $B[v]$ changes.
This concludes the proof of the following theorem:

\begin{theorem}
  It can be tested in linear time whether a c-co-connected clustered
  graph is c-planar if the set of clusters is the union of two
  partitions of the vertex set.
\end{theorem}

Observe that if $(G, \mathcal P_B \cup \mathcal P_R)$ is only
c-connected then $(G, P_B \cup \mathcal P_R)$ does not have to be
c-planar even if $(G,\mathcal P_B)$, $(G,\mathcal P_R)$, and
$(G,\mathcal P'_I)$ are.

\begin{figure}
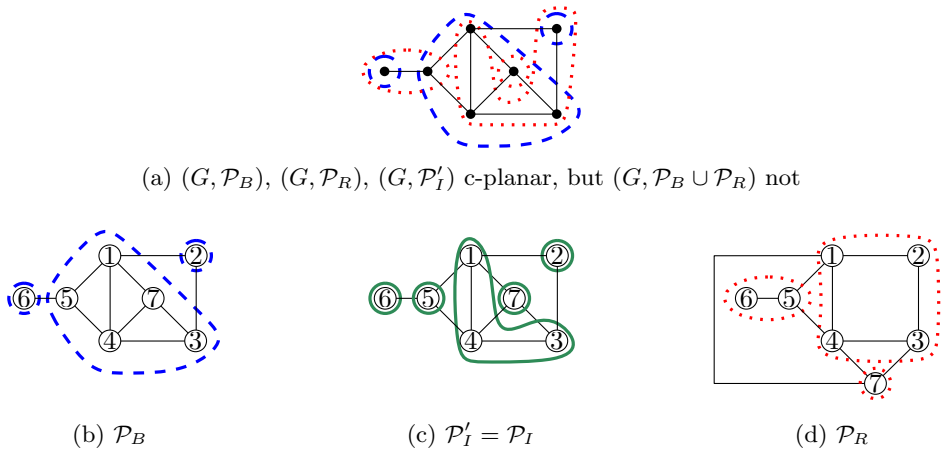

  \begin{minipage}[b]{\linewidth}
     \centering{\includegraphics[page=2]{negative_examples}}%
     \subcaption{$(G,\mathcal P_B)$, $(G,\mathcal P_R)$, $(G,\mathcal P'_I)$ c-planar,
  but $(G, \mathcal P_B \cup \mathcal P_R)$ not}\label{FIG:non_c_planar}
  \end{minipage}\\[1em]
  \begin{minipage}[b]{.33\linewidth}
     \centering{\includegraphics[page=3]{negative_examples}}%
     \subcaption{$\mathcal P_B$}\label{FIG:Proof_non_c_planarB}
  \end{minipage}%
  \begin{minipage}[b]{.33\linewidth}
     \centering{\includegraphics[page=4]{negative_examples}}%
     \subcaption{$\mathcal P_I'=\mathcal P_I$}\label{FIG:Proof_non_c_planarI}
  \end{minipage}%
  \begin{minipage}[b]{.33\linewidth}
     \centering{\includegraphics[page=5]{negative_examples}}%
     \subcaption{$\mathcal P_R$}\label{FIG:Proof_non_c_planarR}
  \end{minipage}%
  \caption{\label{FIG:Proof_non_c_planar}Two c-connected partitions do not have to be c-planar, even if each partition and the intersection partition is.}
\end{figure}

  E.g., let $G$ be the graph in Fig.~\ref{FIG:non_c_planar}, let $\mathcal P_R$ and
  $\mathcal P_B$, respectively, be the partition of the vertex set
  enclosed by the red dotted and blue dashed curves,
  respectively. Let $\mathcal P_I'$ be the connected intersection
  partition of $\mathcal P_R$ and $\mathcal P_B$. 
  %Then $(G,\mathcal P_B)$,
  %$(G,\mathcal P_R)$, and $(G,\mathcal P'_I)$ are c-planar, but $(G,
  %\mathcal P_B \cup \mathcal P_R)$ is not.
  The embedding in Fig.~\ref{FIG:non_c_planar} is c-planar for
  $\mathcal P_B$ and $\mathcal P'_I$~--~see also
  Fig.~\ref{FIG:Proof_non_c_planarB}+\ref{FIG:Proof_non_c_planarI}.
  Fig.~\ref{FIG:Proof_non_c_planarR} shows an embedding that is
  c-planar for $\mathcal P_R$. 
  However, $(G,\mathcal P_B \cup \mathcal P_R)$ is not c-planar:
Assume that there would be an
  embedding that is c-planar for
  $(G, \mathcal P_B \cup \mathcal P_R)$. We use the vertex labeling
  indicated in Fig~\ref{FIG:Proof_non_c_planar}. Due to cluster
  $\{1,2,3,4\}$ the interior of the cycle $c_R=\left<1,2,3,4\right>$
  must be empty. Thus, vertex 5 must be drawn outside $c_R$. Due to
  the cluster $C=\{1,3,4,5,7\}$, vertex 2 and 6 must not be enclosed
  by the triangle $c_B=\left<1,4,5\right>$. It follows that the edges
  connecting 5 to $c_R$ must be drawn such that $c_B$ does not enclose
  $c_R$ and that 6 is outside $c_B$. Due to the edge $\{3,7\}$, vertex
  7 is not enclosed by $c_B$ either. Thus, except for the edge
  $e=\{1,7\}$, the embedding is as indicated in
  Fig.~\ref{FIG:Proof_non_c_planarR}. But no matter how we would add
  $e$ in a planar embedding, we would either create a cycle in $G[C]$
  enclosing vertex 2 or vertex 6.
 
%In the
%next sections, we show how to test c-planarity for c-connected graphs
%with arbitrary overlapping clusters using BC- and SPQR-trees.

\section{C-Connected Clusterings on 2-Connected Graphs}\label{SEC:c-connected-SPQR}

We now describe a polynomial-time method for testing c-planarity for a
planar 2-connected graph $G$ and a set of c-connected clusters
$\mathcal{C}$.  The method described in this section has some
similarities with the algorithm described by Angelini et
al.~\cite{angelini_etal:cg2015} for deciding whether a hierarchically
clustered graph has a drawing with region-region
crossings. The method we give here can, however, be also applied in
Sect.~\ref{SEC:c-connected-BC} to the case where the
underlying graph is not 2-connected.
   
Let $T$ be the
SPQR-tree of $G$ rooted at a $Q$-node $r$ representing the edge $e$ of
$G$. The embeddings of $T$ represent the embeddings of $G$ with $e$ on
the outer face. They in turn induce embeddings of the skeletons of all
nodes of $T$ with their root edges on the outer face.

Apart from the choice of the root, i.e., the choice of the outer face,
the degrees of freedom we have are the order of the parallel components
of the graph at the P-nodes, and in which way the R-nodes are flipped.

In the following let $C \in \mathcal{C}$ be a cluster such that $G[C]$
is connected and let $C_{\text{ext}} \subseteq C$ be a subset of
vertices that we want to be incident to the outer face of $G[C]$
(we'll need $C_{\text{ext}}$ in Sect.~\ref{SEC:c-connected-BC}). We
label each node $\nu$ in $T$ and its corresponding edge in the
skeleton of its predecessor node to capture, which parts of
skel$(\nu)$ are contained in $C$.\footnote{%
  Angelini et al.~\cite{angelini_etal:cg2015} used a similar
  labeling scheme: Their ``full'' corresponds to our
  ``inside'', their ``spined'' corresponds to our
  ``non-outside''. Observe, however, that their labels ``side-spined''
  and ``central-spined'' depend on a given drawing while our labels
  ``double-border'' and ``border'' do not.} Let $s$ and $t$ be the poles of $G_r(\nu)$. The node $\nu$ is an
\emph{inside} node for $C$, if $G_r(\nu)$ is completely contained in
$C$ and at most the poles $s$ and $t$ of $G_r(\nu)$ are in $C_{\text{ext}}$. It is
\emph{inappropriate} if $G_r^-(\nu)$ has no embedding with the poles
on the outer face that is c-planar for $\{C\}$ and is such that
$C_{\text{ext}}$ is on the outer face of $G_r^-(\nu)[C]$. Node $\nu$
is an \emph{outside} node, if it is not inappropriate and $G^-_r(\nu)$
contains no $C$-path between its poles. $\nu$ is \emph{border} if
$\nu$ is neither inside nor outside and $G_r^-(\nu)$ has an embedding
with the poles on the outer face that is c-planar for $\{C\}$ and such
that exactly one of the outer $s$-$t$-paths contains one or more vertices
not in $C \setminus (C_{\text{ext}} \setminus \{s,t\})$. In all other cases $\nu$ is
\emph{double-border}. See Fig.~\ref{FIG:example_labeling} for an
example with $C_{\text{ext}} = \emptyset$.

Traversing the SPQR-tree, we can compute for all nodes whether they
are inside, outside, border, double-border, or inappropriate for a
given cluster $C$. Let node $\nu$ be neither inside nor
outside. If $\nu$ has an inappropriate
child then $\nu$ is also inappropriate. In the following, we assume
that $\nu$ has no inappropriate children.
If $\nu$ is a \textit{P-node} it is \emph{border} if it has no
double-border node and at most one border node among its children. It
is \emph{double-border} if it either has no double-border node and
exactly two border nodes or exactly one double-border node and neither
inside nor border nodes among its children. Otherwise it is
\emph{inappropriate}.
%
% For a child cut vertex $v$ of $H$ let $\nu$ be the node of $T$ such
% that $v$ is a vertex of the skeleton of $\nu$ but not a pole of
% $\nu$. Then $\nu$ is an S-node or an R-node. 
% We call $\nu$ the \emph{top node} of $v$.
%
If $\nu$ is an \emph{S-node}, it is \emph{double-border} if it has at
least one double-border node as child and \emph{border} otherwise.
If $\nu$ is an \emph{R-node} consider skel$^-(\nu)$ embedded with the
poles $s$ and $t$ on the outer face. $\nu$ is \emph{inappropriate} if
there is a simple cycle in skel$^-(\nu)$ that does not contain outside
edges but (a) contains double-border edges or (b) encloses vertices in
$C_{\text{ext}}$ or non-inside edges. Otherwise, it is
\emph{double-border} if both $s$-$t$-paths on the outer face contain
vertices in $C_{\text{ext}}\setminus \{s,t\}$ or non-inside edges. 
In the remaining cases $\nu$ is \emph{border}.
Finally, if $\nu=r$ is the root, let $e$ be the edge of $G$
represented by $r$.  If $e$ does not have both end vertices in $C$ then
$r$ has the same label as its unique child $\mu$. Otherwise $r$ is
\emph{border} if $\mu$ is outer or border and \emph{inappropriate}
otherwise.

\begin{figure}[t!]
\begin{center}
  \begin{minipage}[b]{.45\linewidth}
     \centering{\includegraphics[page=1]{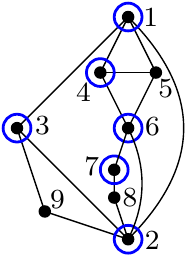}}%
     \subcaption{$(G,\{C\})$}\label{FIG:example_labeling_graph}
  \end{minipage}
  \begin{minipage}[b]{.45\linewidth}
     \centering{\includegraphics[page=2]{example_labeling}}%
     \subcaption{SPQR-tree $T_r$}\label{FIG:example_labeling_tree}
  \end{minipage}
\end{center}
\caption{(a) A graph $G$ with cluster $C=\{1,2,3,4,6,7\}$ containing the circled vertices and (b) the labeling of the nodes of its SPQR-tree: 
thick blue solid circled nodes are inside, dashed nodes border, and dotted nodes double-border. The root is inappropriate.}\label{FIG:example_labeling}
\end{figure}

An \emph{external path} of a node $\nu$ is a path in $G$
between the poles of $G_r(\nu)$ that does not contain any other vertices of
$G_r(\nu)$.  We label the root edge of skel$(\nu)$ \emph{inside} for
$C$ if $\nu$ has an external $C$-path and \emph{outside} otherwise. 
We say that an external path $p$ of node $\nu$ is \emph{to the right
  (left) of} $\nu$ with respect to the ordered pair $(s,t)$ of its
poles if the cycle that is induced by $p$ in the graph that results
from $G$ by contracting $G_r(\nu)$ is oriented (counter-) clockwise 
assuming that
$p$ was oriented from $t$ to $s$. Two external paths of $\nu$ are on
the \emph{same side} of $\nu$ if they are both to the right or both to
the left of $\nu$ with respect to an arbitrary ordering of the poles
of $\nu$. Otherwise, they are on \emph{different sides}.
% %
%The following lemma is a direct consequence of the definition of c-planarity.

%\setcounter{LEMMA:char}{\value{lemma}}

\begin{lemma}\label{LEMMA:char}
  Let $G=(V,E)$ be a 2-connected graph, let $C \subset V$ be a cluster
  inducing a connected subgraph of $G$, let $C_{\text{ext}} \subseteq
  C$, let $T$ the SPQR-tree of $G$, and let $r$ be a Q-node of $T$
  representing the edge $e$ of $G$.  A planar embedding of $G$ with
  $e$ on the outer face is c-planar for $\{C\}$ with $C_{\text{ext}}$
  incident to the outer face of $G[C]$, if and only if the following
  conditions are fulfilled for any non-inside node $\nu$ of its SPQR-tree $T$.
 \begin{enumerate}
 \item All external $C$-paths of $\nu$
   are embedded on the same side of $G_r(\nu)$~--~which we reflect by
   the embedding of the root edge of skel$(\nu)$.
 \item skel$(\nu)$ contains no simple cycle of
   non-outside edges that encloses a non-inside edge or a vertex in $C_{\text{ext}}$.
 \end{enumerate}
\end{lemma}

\begin{proof}
  Clearly, both conditions must be fulfilled for a c-planar embedding
  with $C_{\text{ext}}$ on the outer face of the cluster. So assume
  now that both Conditions  are fulfilled. Let $v \in V \setminus (C
  \setminus C_{\text{ext}})$, let $e' \neq e$ be an edge incident to
  $v$, and let $\nu'$ be the Q-node representing $e'$.  Assume that
  $G$ contains a $C$-cycle $c'$ enclosing $v$.
  By Condition~1, there is no node $\nu$ on the $\nu'$-$r$-path such
  that $c'$ can be decomposed into two external $C$-paths of $\nu$. So,
  let $\nu$ be the first node on the $\nu'$-$r$-path such that $c'$ is
  contained in $G^-_r(\nu)$ or can be composed by a path in
  $G^-_r(\nu)$ and an external $C$-path of $\nu$. Observe that $c'$
  induces a cycle $c$ in skel$(\nu)$ that contains only non-outside
  edges. Let $\mu$ be the child of $\nu$ on the $\nu'$-$\nu$-path.  By
  the choice of $\nu$ it follows that $c$ does not contain the edge
  $e_\mu$ of skel$(\nu)$. Hence $c$ encloses the edge
  $e_\mu$. However, $e_\mu$ was either not inside or $v \in
  C_{\text{ext}}$ is an end vertex of $e_\mu$ not in
  $c$~--~contradicting Condition~2. 
\end{proof}

In the following, we construct a set of binary matrices from an
initial embedding of~$T$ that have the consecutive-ones property, if
and only if there is a c-planar embedding for $\mathcal C$ with the
fixed root edge on the outer face.  The total size of the matrices will be
in $\mathcal O(|V|\ell(\mathcal C))$.

\subsection{Modeling by Consecutive-Ones Property}
For each possible root $r$ of $T$ that is not inappropriate for any $C
\in \mathcal C$, we start with a fixed embedding of $T$~--~including
fixed flips of the R-nodes~--~and perform the following steps:
\subsubsection{Splitting T}
We split $T$ at each R-node,
removing the edges from the R-node to its children from $T$. Let $T_r$
be the subtree containing $r$. For each former non-leaf child $\rho'$
of an R-node $\nu$ we attach a new Q-node $\rho$ to $\rho'$.  We root
the subtree containing $\rho'$ at $\rho$ and denote it by $T_\rho$. We
label $\rho$ inside for a cluster $C$, if $\rho'$ had an external
$C$-path and outside otherwise. In the parent tree, we replace the
R-node $\nu$ by a special P-node $\nu'$ with the same label and three
Q-nodes $\nu_1$, $\nu_2$, $\nu_3$ in this order as children. If the
R-node $\nu$ was labeled border for a cluster $C$, we label $\nu_2$
and exactly one among $\nu_1$ and $\nu_3$ as inside and the other as
outside. More precisely, we label $\nu_1$ as outside if and only if
the left outer path of skel$^-(\nu)$ between its poles contains
non-inside edges or vertices from $C_{\text{ext}}$.  If the R-node $\nu$ was labeled double-border, we
label $\nu_1$ and $\nu_3$ as border and $\nu_2$ as inside.  If the
R-node was labeled inside or outside, we label all three children as
inside or outside, respectively. We thus end up with a forest
containing only S, P, and
Q-nodes. %TODO check for consistency if Q-nodes are labeled anything other than inside or outside

\subsubsection{Initializing the Matrices} For each root $\rho$ of one
of the subtrees, we create a new binary matrix $M_\rho$. A node in
$T_\rho$ is a \emph{lowest-P-child}, if it is the child of a P-node
and has no other P-nodes in its subtree. The embedding of $T_\rho$
induces an ordering of the lowest-P-children from left to right. We
initialize $M_\rho$ with a column for each lowest-P-child in
accordance with the ordering. For a node $\nu$ of $T_\rho$, we use
$c(\nu)$ to refer to the set of its \emph{corresponding columns} in
$M_\rho$, i.e. the columns of the lowest-P-children in $\nu$'s
subtree. For $M_r$ we create one additional \emph{external
  column} $c(r)$. For $\rho \neq r$,
%except the one containing the root $r$ of the original tree, 
we create two additional \emph{external columns}, enclosing the rest
of the matrix. 
For each cluster $C$, one of the two external columns will represent
the side of possible external $C$-paths of the child of $\rho$ and
will be denoted by $c_C(\rho)$.

We then create a row for each non-leaf node $\nu$, adding 1s in the
columns in $c(\nu)$ and 0s in all other columns. This ensures that in
every permutation of the columns of $M_\rho$ for which the 1s are
consecutive in all rows, the columns of the lowest-P-children of each
node remain adjacent, allowing a reconstruction of an embedding of
$T_\rho$ from the ordering of the columns in $M_\rho$. See Fig.~\ref{FIG:matrix_init}. If $\rho \neq
r$ we add two rows having all 1s except for one 0 in the first or last
external column, respectively. 

\begin{figure}[t!]
\begin{center}
  \begin{minipage}[b]{.45\linewidth}
     \centering{\includegraphics[page=1]{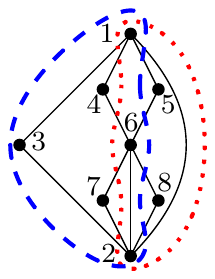}}%
     \subcaption{$(G,\mathcal{C})$}\label{FIG:labels_graph}
  \end{minipage}
  \begin{minipage}[b]{.45\linewidth}
     \centering{\includegraphics[page=2]{example}}%
     \subcaption{SPQR-tree $T_r$}\label{FIG:labels_tree}
  \end{minipage}\\[2ex]
  \begin{minipage}[b]{.55\linewidth}
  	\scalebox{0.7}{\let\quad\thinspace 
$
\setstretch{1.1}
\bordermatrix{
& c(S_1) & c(S_3) & c(S_4) & c(S_5) & c(6,2) & c(S_6) & \underset{(=ext)}{c(1,2)}\cr
P_1 & \tikzmark{LeftInit}1 & 1 & 1 & 1 & 1 & 1 & 0 \cr
S_2 & 0 & 1 & 1 & 1 & 1 & 1 & 0 \cr
P_2 & 0 & 1 & 1 & 0 & 0 & 0 & 0 \cr
P_3 & 0 & 0 & 0 & 1 & 1 & 1 & 0\tikzmark{RightInit} \cr
r_0(P_1,B) & \tikzmark{BleftP1}1 & 0 & 0 & 0 & 0 & 0 & 0 \cr
r_1(P_1,B) & 1 & 1 & 1 & 1 & 1 & 1 & 0 \cr
r_0(P_2,B) & 1 & \tikzmark{BleftS2}\tikzmark{BleftP2}1 & 0\tikzmark{BrightP2} & 0 & 0 & 0 & 0 \cr
r_0(P_3,B) & 1 & 1 & 1 & \tikzmark{BleftP3}1 & 1 & 0\tikzmark{BrightP3}\tikzmark{BrightS2}\tikzmark{BrightP1} & 0 \cr
r_1(P_1,R) & \tikzmark{RleftP1}0 & 1 & 1 & 1 & 1 & 1 & 1 \cr
r_0(P_2,R) & 0 & \tikzmark{RleftS2}\tikzmark{RleftP2}0 & 1\tikzmark{RrightP2} & 1 & 1 & 1 & 1 \cr
r_0(P_3,R) & 0 & 0 & 0 & \tikzmark{RleftP3}0 & 1 & 1\tikzmark{RrightP3}\tikzmark{RrightS2}\tikzmark{RrightP1} & 1 \cr
}
$
\DrawBox[thick, gray, dashdotted]{LeftInit}{RightInit}{\textcolor{gray}{ $\rotatebox[origin=lB]{90}{initialization}$}}
\DrawBox[thick, red, dotted]{RleftP2}{RrightP2}{\textcolor{red}{ $P_2$}}
\DrawBox[thick, red, dotted]{RleftP3}{RrightP3}{\textcolor{red}{ $P_3$}}
\DrawLowBox[thick, red, dotted]{RleftS2}{RrightS2}{\textcolor{red}{ $S_2$}}
\DrawBox[thick, red, dotted]{RleftP1}{RrightP1}{\textcolor{red}{ $P_1$}}
\DrawBox[thick, blue, dashed]{BleftP2}{BrightP2}{\textcolor{blue}{ $P_2$}}
\DrawBox[thick, blue, dashed]{BleftP3}{BrightP3}{\textcolor{blue}{ $P_3$}}
\DrawLowBox[thick, blue, dashed]{BleftS2}{BrightS2}{\textcolor{blue}{ $S_2$}}
\DrawBox[thick, blue, dashed]{BleftP1}{BrightP1}{\textcolor{blue}{ $P_1$}}
}
     \subcaption{$M_r$}\label{FIG:labels_matrix}
  \end{minipage} 
\caption{$\mathcal{C} = \{R ,B\}$ with $R = \{1,2,5,6,8\}$, $B = \{1,2,3,4,6,7\}$. Circled blue nodes in $T_r$ are inside for $B$, squared red nodes inside for $R$. $P_1$, $P_2$, $P_3$, and $S_2$ are border for $R$ and $B$. $R$ and $B$ have a different constraint in $S_2$ and thus different halves of the blocks are filled with 1s in $M_r$.
}
	\label{FIG:matrix_init}
	
\end{center}
\end{figure}

In order to fill the matrix $M_\rho$, we traverse the tree $T_\rho$
with a post-order DFS. For each cluster $C \in \mathcal{C}$ and each
examined node we add up to three rows to $M_\rho$.  We define for each
node $\nu$ and each cluster $C$ a set $r(\nu,C)$ of \emph{relevant
  rows}.   For each lowest P-child $\nu$,
we set $r(\nu,C) = \emptyset$. The block $B(\nu,C)$
is the submatrix of $M_\rho$ with entries in rows $r(\nu,C)$ and
columns $c(\nu)$.
 When we
create rows in $M_\rho$, the default entries are 0 and we explicitly
mention when we set the entries to 1.

\subsubsection{Handling P-nodes}
For a P-node $\nu$ with children $\nu_1, \ldots, \nu_k$ we initialize
$r(\nu,C)$ as $r(\nu_1,C) \cup \ldots \cup r(\nu_k,C)$. Due to
c-planarity, the children of $\nu$ must be permuted such that all
inside children are consecutive pre- and succeeded by at most one
border child and arbitrary many outside children.\footnote{This
  observation was also used by Angelini et
  al.~\cite{angelini_etal:cg2015}. However, they handle distinct
  P-nodes independently while we handle all nodes simultaneously in
  the consecutive ones matrices.} Hence, if $\nu$ is neither outside
nor double-border, we add up to 3 constraint-rows $r_0(\nu, C)$,
$r_1(\nu, C)$, and $r_2(\nu, C)$ to $r(\nu,C)$. If $\nu$ has inside
children, we add $r_0(\nu,C)$
%and place a 1 in $r_0(\nu,C)$ 
with 1s 
in all columns in $c(\nu_i)$ where $\nu_i$
is an inside child of $\nu$. This ensures that all inside children are
placed in consecutive order. If $\nu$ has a child $\mu$ that is a
border node, 
%we copy the 1s from $r_0(\nu, C)$ into a second row
%$r_1(\nu, C)$ and add additional 1s in the columns $c(\mu)$. 
we add $r_1(\nu,C)$ with 1s in all columns in $c(\mu)$ and again with
1s in all columns in $c(\nu_i)$ where $\nu_i$ is an inside child of
$\nu$.
We do the
same for a potential second border node in a third row $r_2(\nu,
C)$. This ensures, that the border children are placed next to the
inside children, with at most one border child on each side.  Finally,
let $\mu$ be a child of $\nu$ let $i \in r(\mu,C)$ and let $j \in
c(\nu) \setminus c(\mu)$. Then we set the entry in row $i$ and column
$j$ to 1, if one or more of the rows in $r(\nu,C) \setminus
r(\mu,C)$ contain a 1 in the same column. See Fig.~\ref{FIG:P-matrix}.
\begin{figure}[t!]
  \begin{center}
  \begin{minipage}[b]{.45\linewidth}
     \centering{\includegraphics[page=1]{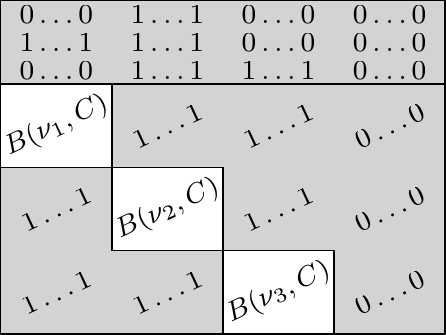}}%
     \subcaption{P-node}\label{FIG:P-matrix}
  \end{minipage}
  \begin{minipage}[b]{.45\linewidth}
     \centering{\includegraphics[page=2]{matrices}}%
     \subcaption{S-node}\label{FIG:S-matrix}
  \end{minipage}
  \caption{(a) The block $B(\nu,C)$ for a P-node $\nu$ with an inside child $\nu_2$, an outside child $\nu_4$ and two border children $\nu_1$ and $\nu_3$. (b) The block $B(\nu,C)$ for an S-node $\nu$ with four children. If $\nu$ is not outer and has an external $C$-path then the upper half, the lower half, or both are filled with 1s.}
   \end{center}
\end{figure}

\subsubsection{Handling S-nodes}
If an S-node $\nu$ with children $\nu_1, \ldots, \nu_k$ is outside then $r(\nu,C) = \emptyset$, 
otherwise $r(\nu,C) = r(\nu_1,C) \cup \ldots \cup r(\nu_k,C)$.

Assume now that $\nu$ is not outside and has an external $C$-path.
Observe that in this case $\nu$ cannot be double-border. Otherwise $r$
would be inappropriate for $C$.  If $\nu$ has two or more P-nodes as
children, we have to make sure that the 1s in each P-node and the 1s
in the external path can be made consecutive via additional 1s.

%then we fill both, the upper and the lower half with one's if $\nu$ is inside and we fill either the upper
%or the lower half with 1s if $\nu$ is border or double-border. We thus make sure that the 1s in each
%P-node and the 1s in the external path can be made consecutive via
%other 1s. 
More precisely, let $\nu_1,\nu_2\dots,\nu_k$ be the children of $\nu$
that are P-nodes. The upper half for a cluster $C$ are all entries in
rows $r(\nu_i,C)$, $i=1,\dots,k$ and columns $c(\nu_i)$, $i =
j+1,\dots,k$ while the lower half are all entries in columns
$c(\nu_i)$, $i=1,\dots,k$ and rows $r(\nu_i,C)$, $i = j+1,\dots,k$. We
fill both, the upper and the lower half with 1s if $\nu$ is inside and
we fill either the upper or the lower half with 1s if $\nu$ is
border. See Fig.~\ref{FIG:S-matrix}.

Recall that if $\nu$ is not inside then the external $C$-paths must
all be on the same side of $G_r(\nu)$ in a c-planar drawing of
$G$. However, external $C_1$- and $C_2$-paths could be on different
sides for distinct clusters $C_1$ and $C_2$. Hence, we cannot just
always fill the upper half with 1s. To this end, 
we will define same and different constraints that only depend on the
structure of the clustered graph. A same (different) constraint
indicates that if there is a c-planar embedding then the external
paths must be on the same (different) side.
%
%
% construct a constraint graph
% that depends only on the structure of $G_r(\nu)$. If the constraint
% graph has no \emph{structurally balanced partition}, i.e., a partition
% of its vertices (i.e.\ the clusters) into two classes $\mathcal
% C_{\textup{upper}}$ and $\mathcal C_{\textup{lower}}$ such that two
% clusters with a same constraint are in the same class and two clusters
% with a different constraint are in different classes, then $G$ has no
% c-planar embedding for the given choice $r$ as the root. Otherwise, we
% assign upper and lower halves to the clusters according to the
% partition of the constraint graph.
% 
%
We call a cluster $C$ \emph{critical} for $\nu$ if $\nu$ is border
with respect to $C$ and has an external $C$-path.

Let $C_1$ and $C_2$ be two clusters that are critical for $\nu$. 
% Assume first that $\nu$ has a child $\nu'$ that is border for both
% $C_1$ and $C_2$.  If $\nu'$ has children that are inside with respect
% to $C_j$, $j=1,2$, let $r_j = r_0(\nu',C_j)$, otherwise let $r_j =
% r_1(\nu',C_j)$. Then $r_j$ has a 0 in one of the columns in
% $c(\nu_i)$. There is a \emph{same constraint} between $C_1$ and $C_2$
% if $r_1$ has a 1 in every column where $r_2$ has a 1 or vice
% versa. There is a \emph{different constraint} between $C_1$ and $C_2$,
% if both, $r_1$ and $r_2$ have both, entries 0 and 1 in $c_1$ and
% $c_2$, and the entry in row $r_1$ and column $c_1$ is 1 if and only if
% the entry in row $r_2$ and column $c_1$ is 0.
%
%
%
% There is a \emph{same constraint} between $C_1$ and $C_2$
% at $\nu$, if there is an $i=1,\dots,k$, a row $r_1 \in r(\nu_i,C_1)$
% and a row in $r_2 \in r(\nu_i,C_2)$ such that at least one among $r_1$
% and $r_2$ has a 0 in one of the columns in $c(\nu_i)$ and such that
% $r_1$ has a 1 in every column where $r_2$ has a 1 or vice verza. There
% is a \emph{different constraint} between $C_1$ and $C_2$, if there is
% a $i=1,\dots,k$, a row $r_1 \in r(\nu_i,C_1)$, a row in $r_2 \in
% r(\nu_i,C_2)$ and two columns in $c_1,c_2 \in c(\nu_i)$ such that
% both, $r_1$ and $r_2$ have both, entries 0 and 1 in $c_1$ and $c_2$,
% and the entry in row $r_1$ and column $c_1$ is 1 if and only if the
% entry in row $r_2$ and column $c_1$ is 0. 
%
If $\nu$ has an external $C_1$-path that is also an external
$C_2$-path then there is a \emph{same constraint} between $C_1$ and
$C_2$. 
% If (1) $\nu$ has neither a child that is border for both $C_1$
% and $C_2$ nor (2) an external $C_1$-path that is also an external
% $C_2$-path then
Otherwise,
there is a \emph{different constraint} between $C_1$ and $C_2$: Assume
that there would be an external $C_1$-path $p_1$ and an external
$C_2$-path $p_2$ of $\nu$ that are on the same side of $G_r(\nu)$ 
in a c-planar embedding of $G$. 
Since $\nu$ is border, there is a $C_i$-path $p_i^\nu$, $i=1,2$ in $G_r^-(\nu)$.
%By
%Cond.~1, $G^-_r(\nu)$ must contain a $C_1$-path between its poles that
%is also a $C_2$-path. 
Consider the cycles $c_i$, $i=1,2$ composed by
$p_i^\nu$ and $p_i$. By c-planarity, each portion of $p_1$ that is inside
$c_2$ must be in $G[C_2]$ and vice versa. Since $p_1$ and
$p_2$ are on the same side of $G_r(\nu)$, there is an 
external path of $\nu$ that contains only edges of $p_1$ inside $c_2$,
edges of $p_2$ inside $c_1$ and common edges of $p_1$ and $p_2$, i.e.,
only edges in $G[C_1 \cap C_2]$.

%Observe that there is at least a same or a different constraint
%between two clusters that are critical for $\nu$ and maybe both.
%
Fix now an arbitrary cluster $C$ that is critical for $\nu$ and assign
$C$ the upper half. Assign to any other cluster $C'$ that is critical
for $\nu$ the upper half if there is a same constraint between $C$ and
$C'$ and the lower half otherwise.

\subsubsection{External Columns}
If $\rho = r$ let $\nu$ be the unique child of $r$ and let $e$ be the
edge represented by the Q-node $r$. Then the external column is $1$
for each row in $r(\nu,C)$ if the cluster $C$ contains both end
vertices of~$e$.

If $\rho \neq r$ then the unique child $\rho'$ of $\rho$ was the child
of an R-node $\nu$. % Let $s'$ and $t'$ be the poles of skel$(\nu)$. 
% and let $e'=\{s',t'\}$ be the edge of skel$(\nu)$ that corresponds to the
% parent edge of $\nu$. 
Consider a fixed embedding of skel$^-(\nu)$ with its poles $s$ and $t$
on the external face. Let $C$ be a cluster for which $\rho'$ is a
border node and has an external $C$-path. We have to make sure that
the parts of $G^-_r(\rho')$ that are not in $C\setminus
C_{\text{ext}}$ are embedded such that they are not enclosed by a
$C$-cycle in $G$ that is composed by an external $C$-path of $\rho'$
and a $C$-path in $G_r(\rho')$ between its poles $s'$ and $t'$.
% Let
% $e$ be the edge of skel$(\nu)$ corresponding to the edge
% $\{\nu,\rho'\}$ of $T$. 
% We consider $e=(s,t)$ to be oriented. Given a
% fixed embedding of $G$, we say that an external $C$-path $p$ of $\nu$ is
% to the right of $e$ if the cycle that is composed by $p$ and an
% $s$-$t$-path in $G_r(\nu)$ is oriented in clockwise
% direction. Otherwise $p$ is to the left of $e$. Observe that in a
% c-planar drawing the external paths of a border or double-border node
% have to be either all right or all left. 

Consider first that skel$^-(\nu)$ contains a cycle $c$ containing
$e_{\rho'}$ and consisting only of non-outside edges. If $c$
is (counter-)clockwise oriented when traversing $e_{\rho'}$ from $s'$
to $t'$, then we set $c_C(\rho)$ to be the (left) right external
column. 

Otherwise all external $C$-paths of $\rho'$ must contain an
external $C$-path of $\nu$. Thus, $\nu$ is not double-border.  
Moreover, the set of vertices of skel$^-(\nu)$ that can be reached
from $s$ using only non-outside edges and not $e_{\rho'}$ 
%Then $t
%\notin S$ and exactly one among $s'$ and $t'$ is in $S$. Hence, $S$
induces an $s$-$t$-cut of skel$^-(\nu)$ that contains $e_{\rho'}$ and
no other non-outside edges. It follows that $e_{\rho'}$ is on the left
(right) outer $s$-$t$-path and all external $C$-paths $\nu$ are to the
left (right) of $\nu$ with respect to $(s,t)$ in any c-planar
embedding. Hence, if $e$ is on the left (right) outer $s$-$t$-path
then we set $c_C(\rho)$ to be the left (right) external column. In
both cases we set the entry in column $c_C(\rho)$ to 1 for each row
in~$r(\rho',C)$.

\bigskip

Clearly the number of columns is linear in the number of Q-nodes and
R-nodes and thus linear in $|V|$ for planar graphs. For a cluster $C$
and a P-node $\nu$, we enter up to three rows but at most if both
poles are in $C$.  Observe that at least one of the poles of a P-node
$\nu$ is not a pole of another P-node $\nu'$ on the path from $\nu$
to the root. Hence, the number of rows is bounded by
$3\ell(\mathcal C)$.

Applying the next theorem with $C_{\text{ext}}=\emptyset$ yields a
characterization of c-connected overlapping clustered graphs with
underlying 2-connected graphs.  
%A proof for the theorem can be found in the appendix. 
%long version~\cite{athenstaedt/cornelsen:arxiv2016}.
 
\setcounter{THM:2-connected}{\value{theorem}}

\begin{theorem}\label{THM:2-connected}
  A c-connected overlapping clustered graph $(G,\mathcal C)$ with an
  underlying planar 2-connected graph $G$ and sets $C_{\text{ext}}
  \subset C$, $C \in \mathcal C$ has a c-planar embedding in which
  $C_{\text{ext}}$ is incident to the outer face of $G[C]$ for any $C
  \in \mathcal C$ if and only if the root of the SPQR-tree of $G$ can
  be chosen such that it is not inappropriate
  for $C \in \mathcal C$ %any cluster
  and
  all matrices $M_\rho$ fulfill the consecutive-ones property.
\end{theorem}

\begin{proof}
  Let $(G,\mathcal C)$ be an overlapping clustered graph.  Let the
  SPQR-tree $T$ of $G$ be rooted at the Q-node $r$, and let $e$ be the
  edge represented by $r$.

  \textbf{Assume first that the columns of all matrices $M_\rho$ are
    permuted such that in each row the 1s are consecutive.}  
  We may assume without loss of generality that the external columns
  were not permuted.
  Starting
  from $\rho=r$, we traverse $T$ and do the following at a non-leaf
  node $\nu$.  If $\nu$ is a P-node, we permute the children
  $\nu_1,\dots,\nu_k$ of $\nu$ according to the ordering of
  $c(\nu_1),\dots,c(\nu_k)$ in the permuted matrix $M_\rho$.
  
  If $\nu$ is an R-node, we fixed an embedding of $G_r(\nu)$ and
  replaced $\nu$ with a P-node and three incident Q-nodes $\nu_1$,
  $\nu_2$, $\nu_3$ in this order.  If $\nu$ was labeled inside or outside
  for all clusters then we maintain the fixed flip of $G_r(\nu)$.
  Otherwise the labeling was such that $c(\nu_2)$ will remain between
  $c(\nu_1)$ and $c(\nu_3)$. We maintain the fixed embedding of
  $G_r(\nu)$ if $c(\nu_1)$ remains before $c(\nu_3)$ after the
  permutation and flip $G_r(\nu)$ otherwise.  If we flip $G_r(\nu)$,
  we also reverse all matrices for all non-leaf nodes in the subtree rooted
  at $\nu$ that are children of an R-node.
  Finally,  we embed $e$ to
  the right of $G^-_r(r)$ if the external column of $M_r$
  is on the right hand side of $M_r$ and to the left otherwise.
  
  We show that this yields a c-planar embedding for $\mathcal C$: Let
  $C \in \mathcal C$ and let $\nu_1$ be a non-inside node of $T$.  We
  show by induction on the length of the $\nu_1$-$r$-path that all
  external $C$-paths of $\nu_1$ are on the same side and that no
  non-inside edge and no vertex in $C_\text{ext}$ is enclosed by a
  simple cycle of non-outside edges
  in skel$(\nu_1)$~--~provided that the root edge of skel$(\nu)$ is
  embedded on the same side as the external $C$-paths of $\nu_1$.
  
\begin{figure}
  \begin{minipage}[b]{.46\linewidth}
     \centering{\includegraphics[page=4]{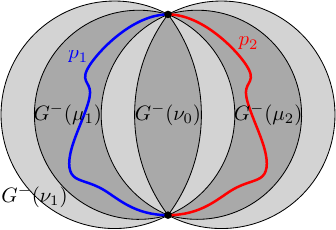}}%
     \subcaption{}\label{FIG:caseA}
  \end{minipage}%
  \begin{minipage}[b]{.26\linewidth}
     \centering{\includegraphics[page=5]{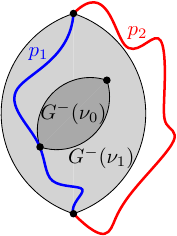}}%
     \subcaption{}\label{FIG:caseB}
  \end{minipage}%
  \begin{minipage}[b]{.28\linewidth}
     \centering{\includegraphics[page=6]{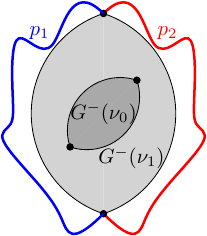}}%
     \subcaption{}\label{FIG:caseC}
  \end{minipage}
  \caption{\label{FIG:threeCases}The three forbidden cases in the
    proof of Theorem~\ref{THM:2-connected}, first direction.}
\end{figure}
  First observe that if $\nu_1$ is an R-node then skel$^-(\nu_1)$ does not
  contain a simple cycle of non-outside edges that encloses a
  non-inside edge, otherwise $\nu_1$ would be inappropriate. 
  Let $e_0$ be a non-inside edge of skel$(\nu_1)$ or let $v_0 \in
  C_{\text{ext}}$ be a vertex of skel$(\nu_1)$ other than the poles
  and let $e_0$ be an edge of skel$(\nu_1)$ incident to $v_0$.  Let
  $\nu_0$ be the child of $\nu_1$ corresponding to $e_0$. Let $p_1$
  and $p_2$ be two paths in $G[C]$ with one of the following
  properties (see Fig.~\ref{FIG:threeCases} for an illustration): (a)
  $\nu_1$ is a P-node and there are two children $\mu_1 \neq \nu_0
  \neq \mu_2$ of $\nu_1$ such that $p_j$, $j=1,2$ is a path in
  $G^-_r(\mu_j)$ between its poles, or (b) $p_1$ is a path in
  $G^-_r(\nu_1)$ between its poles that intersects $G_r(\nu_0)$ at
  most in its poles and $p_2$ is an external $C$-path of $\nu_1$, or
  (c) $p_1$ and $p_2$ are both external $C$-paths of $\nu_1$. We have
  to prove that the cycle composed by $p_1$ and $p_2$ does not enclose
  $G_r(\nu_0)$.

  %\begin{window}[6,r,\includegraphics[page=1]{proof_spqr},] 
  Let $\nu_1,\dots,\nu_\ell=r$ be the $\nu_1$-$r$-path.  Let $j\in\{1,2\}$. If $p_j$ in $G^-(\nu_1)$ let $i_j = 1$. Otherwise
  let $2 \leq i_j \leq \ell$ be minimum such that $\nu_{i_j}$ is an
  R-node or $G^-_r(\nu_{i_j})$ contains $p_j$. We may assume that $i_1
  \leq i_2$. If $\nu_{i_j}$ is an R-node, we actually redefine
  $\nu_{i_j}$ to be the root $\rho$ of the tree containing
  $\nu_{i_j-1}$: we replace $p_j$ by the respective path in $G_r(\nu_{i_j-1})$
  through the root edge $e_\rho$ of skel$(\nu_{i_j-1})$. If
  $\nu_1$ was an R-node we redefine $\nu_1$ to be the special P-node
  with which we replaced the R-node and we redefine $\nu_0$ to be one
  of the artificial non-inside Q-nodes we appended to $\nu_1$.

  Observe that $\nu_{i_j}$ is either $\rho$ or a P-node and $p_j$ is
  composed by two $C$-paths $p_j^1$ and $p_j^2$ connecting the poles
  of $G_r(\nu_1)$ with the poles of $G_r(\nu_{i_j})$ and a middle $C$-path
  $p_j'$. $p_j^1$ and $p_j^2$ are empty if $i_j = 1$. $p_j'$ consists
  of the edge $e_\rho$ if $\nu_{i_j}=\rho$. If $\nu_{i_j}$ is a P-node then
  it has a non-outside child $\mu_j \neq \nu_{i_j - 1}$ such that
  $p'_j$ is a path in $G^-_r(\mu_j)$ between its poles.

  We distinguish some cases. (1) If $\nu_{i_1} = \rho$ or if
  $\nu_{i_1}=\nu_{i_2} \neq \rho$ and $\mu_1 = \mu_2$ then $p_1$ and
  $p_2$ are trivially on the same side of $G_r(\nu_1)$. (2) Assume
  that $\nu_{i_1}=\nu_{i_2} \neq \rho$ and $\mu_1 \neq \mu_2$. Since the
  1s are consecutive in the rows inserted for $\nu_{i_j}$ the two
  non-outside children $\mu_1$ and $\mu_2$ must be on the same side of
  the non-inside child $\nu_{i_j-1}$. (3) Otherwise, observe that the
  $C$-paths $p_2^1$ and $p_2^2$ connecting the poles of $G_r(\nu)$ with the
  poles of $G_r(\nu_{i_2})$ must contain the poles of $G_r(\nu_{k})$,
  $k=1,\dots,{i_2}$. This implies especially that for each
  $k=i_1,\dots,i_2$ the graphs $G^-_r(\nu_k)$ contain a $C$-path
  connecting their poles: such a $C$-path can be composed by $p_1'$
  and portions of $p_2^1$ and $p_2^2$. Hence, $\nu_k$,
  $k=i_1,\dots,i_2$ is not outside. Further a subpath of $p_2$ is an
  external $C$-path of $\nu_{i_1}$. Hence, $\nu_{i_1}$ cannot be
  double-border, since otherwise the root would be inappropriate
  for~$C$.
%\end{window}

\begin{figure}
  \begin{minipage}[b]{.43\linewidth}
     \centering{\includegraphics[page=7]{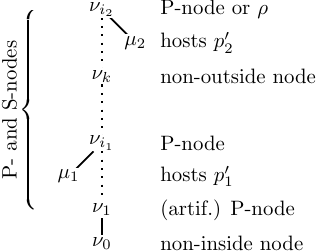}}%
     \subcaption{sub-SPQR-tree}\label{FIG:treeIIIc}
  \end{minipage}%
  \begin{minipage}[b]{.57\linewidth}
     \centering{\includegraphics[page=9]{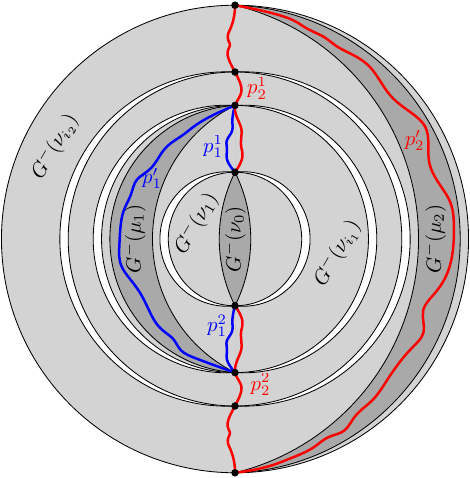}}%
     \subcaption{subgraph}\label{FIG:subgraphIIIc}
  \end{minipage}
\caption{Case (3) in the proof of Theorem~\ref{THM:2-connected}, first direction.}
\end{figure}

  Since $\nu_{i_1-1}$ is non-inside and $\nu_{\mu_1}$ is border or
  inside there is a row $\kappa$ inserted for $\nu_{i_1}$ that
  contains only 0s in $c(\nu_{i_1-1})$ and only 1s in
  $c(\mu_1)$. Further, when we handled $\nu_{i_2}$, we added 1s in the
  row $\kappa$ and the external column (if $\nu_{i_2}=r$) or the
  columns $c(\mu_2)$ (otherwise). Hence, since the 1s must be
  consecutive in $\kappa$, it follows that $c(\nu_{i_1-1})$ cannot be
  between $c(\mu_1)$ and the external column $c_C(\rho)$ or $c(\mu_2)$,
  respectively. Hence, $p_1$ and $p_2$ must be on the same side of
  $G_r(\nu_0)$.

  Now, if $\nu_{i_2}\neq \rho$ we are done. Otherwise let $\nu$ be the
  parent 
  R-node
  of $\nu_{i_2-1}$ in $T$. By induction, we already know that
  all external $C$-paths of $\nu$ are on the same side and that
  $e_{\nu_{i_2-1}}$ is not enclosed by a simple cycle of non-outside
  edges in skel$(\nu)$. Hence, the external $C$-paths of $\nu_{i_2-1}$
  are all on the same side and by construction this is represented by
  the external column $c_C(\rho)$.

  \textbf{Assume now that a c-planar embedding with $e$ on the outer
    face is given in which $C_{\text{ext}}$ is on the outer face of
    $G[C]$ for every cluster $C$.} This yields a permutation of the
  children of the $P$-nodes of $T$ and flips of the R-nodes.  Permute
  the columns of the matrices accordingly. Let $\rho$ be the root of a
  split off tree $T_\rho$ and let $\rho'$ be the only child of $\rho$
  in $T_\rho$.  The external columns of $M_\rho$ are exchanged if on
  the $\rho'$-$r$-path there are an odd number of R-nodes that are
  flipped.

  Recall that we have inserted up to three rows for each P-node and
  each cluster and no other rows into the matrices. Let $\nu$ be a
  P-node in a subtree $T_\rho$ and let $C$ be a cluster such that we
  have created a row $\kappa$ for $\nu$ and $C$ in $M_\rho$. Then
  $\nu$ has no double-border child. Due to c-planarity and the
  condition on all $C_{\text{ext}}$, the children of $\nu$ must be
  permuted such that all inside children are consecutive pre- and
  succeeded by at most one border child and arbitrary many outside
  children. It follows that the 1s in columns $c(\nu)$ must be
  consecutive.

  Let $\nu = \nu_1,\dots,\nu_\ell=\rho$  be the path from
  $\nu$ to the root of $T_\rho$, and let $1 < k \leq \ell$ be maximum
  such that $\nu_1,\dots,\nu_k$ are not outside. 
  If $\nu_1$ was not a special $P$-node substituting an R-node then 
  $\nu_i$
  is a P-node if $i$ is odd and an $S$-node if $i$ is even. 
  (Otherwise it might be vice versa, but the situation is similar) 
  $\nu_k$ is
  a P-node if $\nu_k \neq \rho$. Also observe that $c(\nu_{i-1})
  \subseteq c(\nu_i)$, $i=2,\dots,\ell$ and that for each
  $i=1,\dots,\ell$ the columns in $c(\nu_i)$ are consecutive in the
  permuted matrix. If $k < \ell$, we've set $r(\nu_{k+1},C) =
  \emptyset$. Hence, the entries in row $\kappa$ are 0 in all columns
  in $c(\nu_{\ell-1}) \setminus c(\nu_k)$.

  \begin{window}[4,r,\includegraphics[page=8]{proof_nu1},] 
  \indent
  We consider first a P-node $\nu_i$, $i=3,\dots,k$ odd.  Since
  $\nu_{i-1}$ is not outside it follows that no child of $\nu_i$ other
  than $\nu_{i-1}$ can be double-border. Hence, for each non-outside
  child $\mu \neq \nu_{i-1}$ of $\nu_i$ there are 1s in row $\kappa$
  and all columns in 
  %$c(\nu_{i-1}) \setminus c(\nu_{i-2})$. 
  %
  $c(\mu)$.
  Observe
  that due to c-planarity the non-outside children of $\nu_i$ are
  consecutive. Moreover, if there are both, non-outside children of
  $\nu_i$ to the right and the left of $\nu_{i-1}$ then $\nu_{i-1}$ is
  inside and, thus, all columns in $c(\nu_1)$ as well as
  $c(\nu_{j}) \setminus c(\nu_{j-1})$ have entry 1 in row $\kappa$
  for all $3 \leq j < i$ odd. 

  \indent
  If $\nu_k = \rho$, let $m = k = \ell$ and assume that the external
  $C$-paths of $\nu_{\ell-1}$ are all to the right (left) of
  $\nu_{\ell-1}$, i.e., the column $c_C(\rho)$ is the right (left)
  external column.
  If $\nu_k \neq \rho$, let $m \leq k$ be maximum such that $\nu_m$ is
  a P-node and has a non-outside child other than $\nu_{m-1}$ (If no
  such P-node exists then all entries in row $\kappa$ other than in
  the columns $c(\nu)$ are zero and thus all 1s are consecutive.)
  Assume that $\nu_m$ has a non-outside child $\mu$ to the right
  (left) of $\nu_{m-1}$.  Assume now that
  %there is a non-outside
  %child $\mu \neq \nu_{m-1}$ to the right (left) of $\nu_{m-1}$ and
  there is a $1 \leq j < m$ odd such that the P-node $\nu_j$ has a
  child $\mu'$ to the right (left) of $\nu_{j-1}$. I.e., the columns
  $c(\mu')$ are between the columns $c(\nu)$ and $c(\mu)$. If $\mu'$
  were not inside then $G_r(\mu')$ would contain a vertex in $V
  \setminus (C \setminus C_{\text{ext}})$ that would be enclosed by a
  $C$-cycle composed by the following four paths: (1) A $C$-path in
  $G^-_r(\nu_{j-1})$ between its poles, (2+3) two $C$-paths connecting
  the poles of $G_r(\nu_{m-1})$ with the poles of $G_r(\nu_j)$, and
  (4) either an external $C$-path of $\nu_{\ell-1}$, if $\nu_{m}=\rho$
  or a $C$-path in $G_r(\mu)$ between its poles, if $\nu_m$ is a
  P-node. Hence, the entries in $c(\mu')$ are all 1.
  \end{window}

  Consider now an S-node $\nu_i$, $i=2,\dots,k$ even that has an
  external $C$-path. By the choice of $k$, $\nu_i$ is not
  outside. Since the root is not inappropriate, $\nu_i$ is not
  double-border. Thus, we've set the entries in row $\kappa$ and columns
  $c(\nu_i) \setminus c(\nu_{i-1})$ to 1 if $\nu_i$ is
  inside. Otherwise, we set the entries in $c(\nu_i) \setminus
  c(\nu_{i-1})$ that are to one side of $c(\nu_{i-1})$ to 1. Observe
  that an $S$-node $\nu_i$ has an external path if and only if
  $\ell=k$ or 
  $i < m$. 
  %there is a P-node $\nu_j$, $i < j \leq k$ that has a
  %non-outside child other than $\nu_{j-1}$.

  Hence, row $\kappa$ looks as follows.  Assume without loss of
  generality that $c(\nu_m)\setminus c(\nu_{m-1})$ contains a 1 to
  the right of $c(\nu_{m-1})$. Then the entries in $c(\nu)$
  are ordered such that all 0s (if any) are to the left and all 1s are
  to the right. Moreover, if $\nu$ is inside let $1 \leq b \leq k$ be
  maximal such that $\nu_b$ is inside. Then all entries in columns
  $c(\nu_b)$ are 1. Otherwise let $b=1$.
  % and set $c(\nu_0,C)$ be the union of all sets $c(\mu,C)$ for all
  % non-outside children $\mu$ of $\nu$. 
  For $i=b+1,\dots,m-2$ odd all entries in $c(\nu_i) \setminus
  c(\nu_{i-1})$ that are on the right side of $c(\nu_{i-1})$ are
  1. For $i=b+1,\dots,m-1$ even, all entries in $c(\nu_i) \setminus
  c(\nu_{i-1})$ on one side of $c(\nu_{i-1})$ are 1~--~however, for
  some $i$ that could be the right-hand side and for others the
  left-hand side. Finally, the entries in $c(\nu_m) \setminus
  c(\nu_{m-1})$ to the right of $c(\nu_{m-1})$ are ordered such
  that the 1s are to the left and the 0s (if any) are to the right.
  See Fig.~\ref{FIG:cons_ones}.

  \begin{figure}
    \begin{center}
      {\includegraphics[width=\textwidth,page=3]{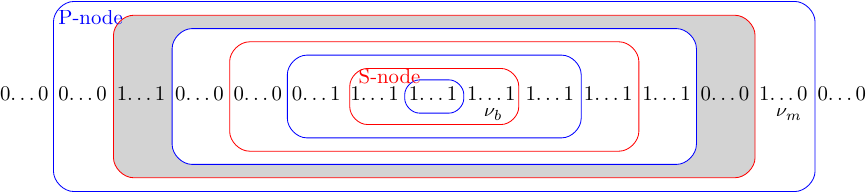}}
    \end{center}
    \caption{\label{FIG:cons_ones}Possible ordering of the 0s and 1s in a row inserted for a P-node $\nu_1$ according to a c-planar embedding.}
  \end{figure}
  Hence, the 1s in row $\kappa$ are consecutive up to maybe the wrong
  choice of the side for the 1s inserted for border S-nodes.  Observe,
  however, on one hand that we could remove now the 1s from the wrong
  side and insert them on the right side and would thus obtain the 1s
  consecutive. We could obtain that for one cluster also by permuting
  the columns for the children of the S-node accordingly. On the other
  hand the assignment to sides was forced by the same and different
  constraints~--~up to the choice for one cluster. Hence, if we do
  the permuting that works for one cluster it'll create the feasible
  assignment we'd obtain if we'd assign the sides now that we knew were
  the external paths are embedded.
\end{proof}

% \begin{figure}
% 	\subfloat[outside]{\includegraphics[page=1]{R_node}
% 		\label{FIG:R_node_outside}
% 	}
% 	\hfill
% 	\subfloat[inside]{\includegraphics[page=2]{R_node}
% 		\label{FIG:R_node_inside}	
% 	}
% 	\hfill
% 	\subfloat[border]{\includegraphics[page=3]{R_node}
% 		\label{FIG:R_node_border1}
% 	}
% 	\hfill	
% 	\subfloat[border]{\includegraphics[page=4]{R_node}
% 		\label{FIG:R_node_border2}
% 	}
% 	\hfill
% 	\subfloat[double border]{\includegraphics[page=5]{R_node}
% 		\label{FIG:R_node_double_border}
% 	}
% 	\hfill
% 	\subfloat[reject]{\includegraphics[page=6]{R_node}
% 		\label{FIG:R_node_reject}
% 	}
% 	\caption{R-node}
% 	\label{FIG:R_node}
% \end{figure}

% \begin{figure}
% \subfloat[]{\includegraphics{matrix_from_P_node}}
% \subfloat[]{\begin{tabular}[b]{|c|c|c|c|c|c|}
% \hline 
% & \multicolumn{5}{c|}{$P_1$} \\ 
% \hline 
% & $N_1$ & \multicolumn{2}{c|}{$P_2$} & $N_4$ & $N_5$ \\ 
% \hline 
% & $N_1$ & $N_2$ & $N_3$ & $N_4$ & $N_5$ \\ 
% \hline 
% \multirow{3}{*}{$P_1$} & $1\ldots 1$  & \multicolumn{2}{c|}{$0 \ldots \ldots 0$} & $1\ldots 1$  & $0 \ldots 0$ \\ 
%  & $1\ldots 1$  & \multicolumn{2}{c|}{$1\ldots \ldots 1$}  & $1\ldots 1$  & $0 \ldots 0$ \\ 
%  & $1\ldots 1$  & \multicolumn{2}{c|}{$0 \ldots \ldots 0$} & $1\ldots 1$  & $1\ldots 1$  \\ 
% \hline 
% $P_2$ & $1\ldots 1$  & $1\ldots 1$  & $0 \ldots 0$ & $1\ldots 1$  & $1\ldots 1$  \\ 
% \hline 
% \end{tabular} 
% }
% 	\caption{Part of an SPQR-tree with two P-nodes and five arbitrary other nodes. A full circle indicates an inside-node, a dashed circle a border node.}
% 	\label{FIG:matrix_from_P_node}
% \end{figure}

\section{C-Connected Clusterings on Arbitrary Graphs}\label{SEC:c-connected-BC}
Let $(G,\mathcal C)$ be a c-connected overlapping clustered graph with
underlying planar graph $G$.
We show how to extend the method from the last section to work for an
arbitrary planar graph $G$. If $G$ is not connected, we can test each
connected component separately, since the c-connectivity limits each
cluster to a single component.

It remains the case, where $G$ is connected but not 2-connected and
can thus be represented by a BC-tree. We consider the BC-tree of $G$
rooted at a block $H_r$ (meaning that $H_r$ should contain an edge
incident to the outer face of $G$ in a planar drawing). Let $H$ be a
block of $G$. If $H \neq H_r$ then the \emph{parent cut vertex} of $H$
is the cut vertex of $H$ on the path from $H$ to $H_r$.  $H$ 
is a \emph{child block} of its parent cut vertex. All other cut
vertices of $H$ are called \emph{child cut vertices} of $H$. All
cut vertices of $H_r$ are \emph{child cut vertices} of $H_r$.

Consider the SPQR-tree $T$ of $H$. If $H = H_r$, any root of $T$ is
\emph{suitable}. Otherwise the parent cut vertex of $H$ must be on the
outer face of $H$. Thus, a root of the SPQR-tree $T$ is
\emph{suitable} if it corresponds to an edge incident to the parent
cut vertex of $H$. See Fig.~\ref{FIG:BC-tree}.

Let $H_1,\dots,H_k$ be the child blocks of $v$ and let $V_i$,
$i=1,\dots,k$ be the set of vertices in the connected components of
$G-v$ containing $H_i$.  We call a cluster $C$ \emph{relevant} for a
child block $H_i$, if $v \in C$ and $V_i \not \subseteq C$. Let
$C_{\text{ext}}$ be the set of child cut vertices $v$ of $H$ such that $C$
is relevant for a child block of $v$.

% Now when recursively labeling the nodes of $T$ from the leaves to the
% root, we do the following two modifications.  For each child cut
% vertex $v$ of $H$ and for each cluster $C$ that is relevant for some
% child block of $v$
% \begin{itemize}\label{ITEM:modified_labels}
% \item if the top node $\nu$ of $v$ is an S-node that would be labeled
%   inside then label $\nu$ border.
% \item if the top node $\nu$ of $v$ is an R-node and $v$ is on one of
%   the two paths on the external face between the poles then label
%   $\nu$ as if that path would contain a border edge.
% \end{itemize}

Use the algorithm for 2-connected graphs, restricting the
roots for the SPQR-trees to be suitable, to test whether there is some c-planar
embedding for each block $H$ with the parent cut vertex on the outer
face of $H$ and $C_{\text{ext}}$ on the outer face of $H[C]$.
For each child cut vertex $v$ of a block $H$ and for each child block
$H_i$ of $v$, test whether there is a free face, i.e., a face $f$ of $H$ incident to $v$
such that the boundary of $f$ contains a vertex not in $C$ for any
cluster $C$ that is relevant for $H_i$. If so, the c-planar embeddings
of the blocks can be combined into a c-planar embedding of the whole
graph. In the following, we show that otherwise there is no c-planar
drawing for the whole graph with the given choices of the root of the
BC-tree and the roots of the SPQR-trees. 
%The detailed proofs of the
%statements can be found in the appendix.

\begin{figure}[t!]
  \begin{minipage}[b]{.33\linewidth}
     \centering{\includegraphics[page=2]{BC-tree}}%
     \subcaption{rooted BC-tree}
  \end{minipage}
    \begin{minipage}[b]{.45\linewidth}
     \centering{\includegraphics[page=1]{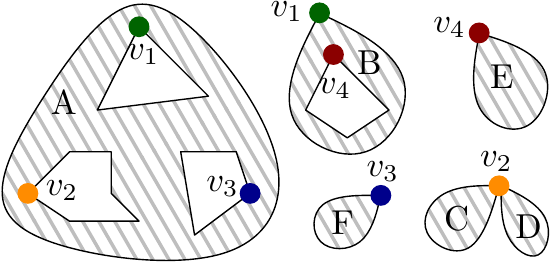}}%
     \subcaption{the corresponding embedding}
  \end{minipage}
	\caption{\label{FIG:BC-tree}The hierarchy in the BC-tree given by the choice of the root}
 %induces a possible embedding: the parent of a cut vertex can be embedded in a way that the cut vertex is placed on an inner face, all its children must place the cut vertex on the outer face}
\end{figure}

Given a c-planar embedding, a face $f$ is \emph{free} with respect to
a subset $\mathcal{C}' \subseteq \mathcal{C}$ of clusters if $f$ is
not enclosed by a $C$-cycle for any $C \in \mathcal{C}'$. Otherwise,
$f$ is \emph{covered} by $\mathcal{C}'$. In the following we use $C$
instead of $\{C\}$ if the context is clear. 

\begin{remark}
A face is covered by $C$ if and only if its boundary is a $C$-cycle.
\end{remark}

A vertex $v$ is
\emph{free} with respect to a subset $\mathcal{C}' \subseteq
\mathcal{C}$ of clusters if one of its incident faces is free with
respect to $\mathcal{C}'$ and \emph{covered} by $\mathcal{C}'$ otherwise. We
call a cut vertex $v$ \emph{free} for a child block $H_i$, if $v$ is
free with respect to the set of clusters that are relevant for $H_i$.

Given a vertex $v$ in a block $H$, we call two incident edges $e_1$
and $e_2$ of $v$ \emph{equivalent} with respect to a set
$\mathcal{C}'$ of clusters, if they are in the same block of
$H[\bigcap_{C \in \mathcal{C}'} C]$, i.e. if there is a simple cycle
in $H$ that is a $C$-cycle for any $C \in \mathcal{C}'$ and contains
both, $e_1$ and $e_2$.
A \emph{$\mathcal C'$-equivalence class} around $v$ is a maximal set
of edges incident to $v$ that are pairwise equivalent with respect to
$\mathcal C'$.

\begin{lemma}\label{LEMMA:free_vertex_one}
  Let $v$ be a vertex of a block $H$ and let $\mathcal C' \subseteq
  \mathcal C$. Then a $\mathcal C'$-equivalence class around $v$ is 
  a consecutive set in the cyclic order around $v$ in
  any c-planar embedding of~$H$.
\end{lemma}

\begin{proof}
  Let $\mathcal{C}'$ be a set of clusters, $v$ a vertex in block $H$
  and let $e_i=\{v,v_i\}$, $i=1,2$ be two edges incident to $v$ that
  are contained in a simple cycle $c$ in
  $H[\bigcap_{C \in \mathcal{C}'} C]$.  See
  Fig.~\ref{FIG:lemma5}. Then all vertices that are enclosed by $c$
  are in $\bigcap_{C \in \mathcal{C}'} C$. Let $e'=\{v,v'\}$ be an
  edge enclosed by $c$. Let $i \in \{1,2\}$.  Since $H$ is
  2-connected, there must be a $v'$-$v_i$ path $p_i$ in $H$ not
  containing $v$. Let $v_i'$ be the first vertex of $p_i$ on $c$. Let
  $c_i'$ be the cycle formed by the $v'$-$v_i'$-subpath of $p_i$, the
  $v_i'$-$v$-subpath of $c$ containing $e_i$ and the edge $e'$. Then
  $c_i'$ is a simple cycle in $H[\bigcap_{C \in \mathcal{C}'} C]$
  containing $e'$ and $e_i$.
\end{proof}

Let $v$ be a vertex that is free with respect to any $C \in
\mathcal{C}'$. Given a c-planar embedding of a block, a
\emph{$\mathcal{C}'$-interval} around a vertex $v$ is a maximal
sequence of consecutive edges around~$v$
%\begin{itemize}
%\item 
that are (a) equivalent with respect to $\mathcal C'$ and  such that
%\item 
(b)
the face between any two consecutive edges is covered by
  $C$ for all $C \in \mathcal{C}'$.
%\end{itemize}
Note that there is a one-to-one correspondence between the $\mathcal
C'$-equivalence classes and the $\mathcal{C}'$-intervals around $v$: the
condition that $v$ is free with respect to any cluster in
$\mathcal{C}'$ guarantees that the $\mathcal{C}'$-intervals have a
well defined start and end point.  Also note that there might be
several distinct $\mathcal{C}'$-intervals around~$v$~--~even if
$\mathcal{C}'$ contains only one cluster.

\begin{lemma}\label{LEMMA:free_vertex}
Let $H$ be a block, $v$ a vertex in $H$, and $\mathcal{C}' \subseteq  
% a subset of the set  
\mathcal{C}$.
% of clusters.  
If there is a c-planar
embedding of $H$, in which $v$ is free with respect to $\mathcal{C}'$,
then $v$ is free with respect to $\mathcal{C}'$ in any c-planar
embedding of $H$ in which $v$ is free with respect to $C$ for all $C
\in \mathcal{C}'$.
\end{lemma}
\begin{proof}
  Assume that there is a c-planar embedding of $H$ in which $v$ is
  free with respect to $C$ for all $C \in \mathcal{C}'$ but $v$ is not
  free with respect to $\mathcal C'$. Consider the cyclic order
  $e_1, \dots e_\ell$ of the edges around $v$. Since $v$ is not
  covered by any $C \in \mathcal{C}'$, the $C$-intervals around $v$
  are well defined.  Among all $C$-intervals for all
  $C \in \mathcal{C}'$, let $\mathcal{I}$ be a minimal set of
  intervals such that all faces around $v$ are covered by at least one
  interval in $\mathcal{I}$. Let
  $I_i = \left<e_{s_i}, \ldots, e_{t_i}\right>, i = 1,\ldots,\kappa$
  be the intervals in $\mathcal{I}$ in cyclic order around $v$. See
  Fig.~\ref{FIG:lemma6}. We assume that $s_1 = 1$, $s_i < t_i$ for
  $i = 1, \ldots, \kappa -1$, and $t_\kappa > s_\kappa$. For
  simplicity, we set $s_{\kappa + 1} := s_1$. Let
  $C_i \in \mathcal{C}'$ be such that $I_i$ is a $C_i$-interval.
  Since all faces around $v$ are covered, it holds that
  $s_{i+1} \leq t_{i}$.
  $\left< e_{s_{i+1}}, \ldots, e_{t_{i}} \right>$ is a
  $\{C_i, C_{i+1}\}$-interval (Let $s_{i+1} \leq j < t_i$. The face
  $f$ between $e_j$ and $e_{j+1}$ is covered by both, $C_i$ and
  $C_{i+1}$. Hence the boundary of $f$ is a both a $C_i$- and a
  $C_{i+1}$-cycle). Thus $\{ e_{s_{i+1}}, \ldots, e_{t_{i}} \}$ is
  consecutive in any c-planar embedding. Hence, in any c-planar
  embedding, the ordering of edges around $v$ is as
  follows\\[1em]

\begin{figure}
  \begin{minipage}[b]{.45\linewidth}
     \centering{\includegraphics[page=1]{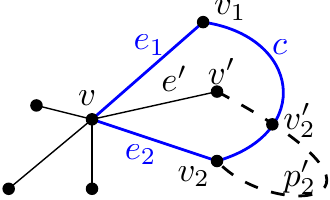}}%
     \subcaption{$e_1 \sim e_2 \Rightarrow e_1 \sim e' \sim e_2$}\label{FIG:lemma5}
  \end{minipage}
  \begin{minipage}[b]{.45\linewidth}
     \centering{\includegraphics[page=2]{proof_general}}%
     \subcaption{Minimal set of intervals covering $v$}\label{FIG:lemma6}
  \end{minipage}
\caption{
  Illustration of the proofs of 
  (a) Lemma~\ref{LEMMA:free_vertex_one} and
  (b) Lemma~\ref{LEMMA:free_vertex}. (The symbol $\sim$ refers to equivalence with respect to the set $\mathcal C'$.)
}
\end{figure}

  %\pagebreak

  \noindent
  $ \displaystyle \overunderbraces { &\br{3}{C_1-\text{interval}} } {
    & \{e_1,\dots,e_{t_\kappa}\},&
    \{e_{t_\kappa+1},\dots,e_{s_2-1}\},& \{e_{s_2},\dots,e_{t_1}\},&
    \{e_{t_{1}+1},\dots,e_{s_{3}-1}\},& \{e_{s_{3}},\dots,e_{t_2}\},&
  } { &&&\br{3}{C_2-\text{interval}} }
  \dots$\\[1em]
  \centerline{ $\dots, \{e_{t_{i-1}+1},\dots,e_{s_{i+1}-1}\},
    \{e_{s_{i+1}},\dots,e_{t_i}\}, \dots$
  }\\[1em]
  \phantom{x} \hfill $\overunderbraces
  {
    &\br{3}{C_{\kappa-1}-\text{interval}}    
  }
  {
    &\dots,&
    \{e_{t_{\kappa-2}+1},\dots,e_{s_{\kappa}-1}\},& 
    \{e_{s_{\kappa}},\dots,e_{t_{\kappa-1}}\},&
    \{e_{t_{\kappa-1}+1},\dots,e_{\ell}\},&
    \{e_1,\dots,e_{t_\kappa}\}
  }
  {
    &&&\br{3}{C_\kappa-\text{interval}}
  }
  $\\[2em]
  and thus, a sequence of overlapping $C$-intervals for some $C \in
  \mathcal C'$. Hence, for any face $f$ incident to $v$ there is at
  least one $C \in \mathcal C'$ such that $f$ is covered by $C$.
  Therefore $v$ cannot be free with respect to $\mathcal C'$ in any
  c-planar embedding.
\end{proof}

We now apply Lemma~\ref{LEMMA:free_vertex} to any child cut vertex $v$
of any block $H$ and to the set $\mathcal C'$ of relevant clusters of
any child block of $v$ to obtain our main result. Observe that the
particular choice of $C_{\text{ext}}$ in the following theorem
guarantees that the child cut vertices are free with respect to any
relevant cluster.

\setcounter{THM:arbitrary}{\value{theorem}}
\begin{theorem}
  A c-connected overlapping clustered graph $(G,\mathcal{C})$ is
  c-planar, if and only if $G$ is planar and for each connected
  component of $G$, there is a root block of its BC-tree for which
  there exist suitable root nodes of the SPQR-tree of each block that
  are not inappropriate for any $C
  \in \mathcal{C}$ with 
  $
  C_{\text{ext}} = 
  \{v;\; v\text{ child cut vertex and } C 
  \text{ relevant for a child block of } v \}
  $ 
  such that
\begin{enumerate}
\item all binary matrices 
  fulfill the consecutive-ones property and
\item given an arbitrary consecutive-ones ordering of the binary matrices each
  cut vertex is free for each of its child blocks in the corresponding
  embedding.
\end{enumerate}
\end{theorem}
\begin{proof}

  Assume the two conditions hold. We embed the blocks as in the proof
  of Theorem~\ref{THM:2-connected} and combine the embeddings of the
  blocks as follows. Let $H$ be a block, let $v$ be a child cut vertex
  of $H$ and let $H_i$ be a child block of $v$. We place $H_i$ into a
  face of $H$ incident to $v$ that is free with respect to the set of
  $H_i$'s relevant clusters. This yields a c-planar embedding of $G$:

  Otherwise there must be a cluster $C$ and a vertex $w \in V
  \setminus C$ such that $w$ is enclosed by a
  $C$-cycle $c$. Let $c$ be in block $H$. The first condition
  requires, that the embedding of $H$ is c-planar (see
  Theorem~\ref{THM:2-connected}). Hence, $w$ cannot be a vertex of $H$.
  Let $v'$ be the parent cut vertex of $H$ and let $V'$ be the union
  of the sets of vertices in the connected components of $G-v'$ not
  containing $H$. By the choice of the root of the BC-tree, $V'$ must
  be drawn in the outer face of $H$. Hence $w \notin V'$. 

  Finally, let $v$ be a child cut vertex of $H$, let $H_i$ be a child
  block of $v$, let $V_i$ be the set of vertices in the connected
  components of $G-v$ containing $H_i$, and assume that $w \in V_i$.
  Then $v$ must be enclosed by $c$ and thus, by c-planarity of
  $H$, $v \in C$. Since $w \notin C$ it follows that $C$ is relevant
  for $H_i$.  Since we embedded $H_i$ into a face of $H$ that was free
  with respect to $H_i$'s relevant clusters, it follows that $w$
  cannot be enclosed by the $C$-cycle $c$.

  For the other direction assume now that there is a c-planar
  embedding $\mathcal E$.  Without loss of generality, we assume that
  $G$ is connected. Let the root $H_r$ of the BC-tree be a block with
  an edge that is incident to the outer face of $G$. Root each
  SPQR-tree at an edge incident to the outer face of the respective
  block and incident to the parent cut vertex.

  In a c-planar drawing, a child cut vertex $v$ of a block $H$ is
  placed on the outer face of $H[C]$ for any relevant cluster $C$ of
  any of $v$'s child blocks. Thus, Theorem~\ref{THM:2-connected}
  implies that the roots are not inappropriate and Condition~1 is
  fulfilled.

  Obviously any block must be inserted into a face that is free with
  respect to its relevant clusters in any c-planar embedding of $G$.
  Consider now a block $H$ and an embedding $\mathcal E'$ of $H$
  corresponding to an arbitrary consecutive-ones ordering of the
  binary matrices. Let $v$ be a child cut vertex of $H$ and let $H_i$
  be a child block of $v$. Let $\mathcal C'$ be the set of relevant
  clusters for $H_i$. 

  The labeling guarantees that $v$ is on the outer face of $H[C]$ for any $C \in
  \mathcal C'$. Thus,  $\mathcal E'$ is a c-planar
  embedding of $H$ in which $v$ is free with respect to each $C \in
  \mathcal C'$.
  We further know that $\mathcal E$ induces a
  c-planar embedding of $H$ in which $v$ is free with respect to
  $\mathcal C'$. 
  Hence, Lemma~\ref{LEMMA:free_vertex} implies that $v$ is free with
  respect to $\mathcal C'$ in $\mathcal E'$. 
\end{proof}

The characterization in the previous theorem immediately yields the
following corollary.

\begin{corollary}
  It can be tested in polynomial time whether a c-connected
  overlapping clustered graph is c-planar.
\end{corollary}


\begin{thebibliography}{10}

\bibitem{angelini_etal:cg2015}
P.~Angelini, G.~Da Lozzo, G.~Di Battista, F.~Frati, M.~Patrignani, and V.~Roselli.
\newblock Relaxing the constraints of clustered planarity.
\newblock {\em Computational Geometry: Theory and Applications}, 48:42--75, 2015.

%\bibitem{athenstaedt/cornelsen:arxiv2016}
%J.~C. Athenst{\"a}dt and S.~Cornelsen.
%\newblock Overlapping c-planarity and the union of two partitions.
%\newblock ArXiv, 2016.

\bibitem{athenstaedt_etal:gd2014}
J.~C.~Athenst{\"a}dt, T.~Hartmann, and M.~N{\"o}llenburg.
\newblock Simultaneous embeddability of two partitions.
\newblock In C.~Duncan and A.~Symvonis, editors, {\em GD~2014}, LNCS, vol. 8871, pp. 64--75. Springer, 2014.

\bibitem{bixbi/wagner:1988}
R.~E.~Bixbi and D.~K.~Wagner.
\newblock An almost linear time algorithm for graph realization.
\newblock {\em Mathematics of Operations Research}, 13(1):99--122, 1988.

\bibitem{vanBeveren_etal:gd2016}
R.~van Beveren, I.~A.~Kanj, C.~Komusiewicz, R.~Niedermeier, and M.~Sorge.
\newblock Twins in subdivision drawings of hypergraphs.
\newblock In Y.~Hu and M.~N{\"o}llenburg, editors, {\em GD~2016}, LNCS 9801. Springer, to appear.

\bibitem{blaesius/rutter:2016}
T.~Bl{\"a}sius and I.~Rutter.
\newblock A new perspective on clustered planarity as a combinatorial embedding
  problem.
\newblock {\em Theoretical Computer Science}, 609(P2):306--315, 2016.

\bibitem{booth/lueker:1976}
K.~S.~Booth and G.~S.~Lueker.
\newblock Testing for the consecutives ones property, interval graphs, and
  graph planarity using {PQ}-tree algorithms.
\newblock {\em Journal of Computer and System Sciences}, 13:335--379, 1976.

\bibitem{brandes_etal:iwoca10outer}
U.~Brandes, S.~Cornelsen, B.~Pampel, and A.~Sallaberry.
\newblock Blocks of hypergraphs applied to hypergraphs and outerplanarity.
\newblock In C.~Iliopoulos and W.~Smyth, editors, {\em IWOCA~2010}, LNCS vol. 6460, pp. 201--211. Springer, 2011.

\bibitem{brandes_etal:2011path}
U.~Brandes, S.~Cornelsen, B.~Pampel, and A.~Sallaberry.
\newblock Path-based supports for hypergraphs.
\newblock {\em Journal of Discrete Algorithms}, 14:248--261, 2011.

\bibitem{buchin_etal:2011}
K.~Buchin, M.~{van Kreveld}, H.~Meijer, B.~Speckmann, and K.~Verbeek.
\newblock On planar supports for hypergraphs.
\newblock {\em Journal on Graph Algorithms and Applications}, 15(1), 2011.

\bibitem{cornelsen/wagner:wg2003}
S.~Cornelsen and D.~Wagner.
\newblock Completely connected clustered graphs.
\newblock In H.~L. Bodlaender, editor, {\em WG~2003}, LNCS vol. 2880, pp. 168--179. Springer, 2003.

\bibitem{cortese_etal:2008}
P.~F. Cortese, G.~{Di Battista}, F.~Frati, M.~Patrignani, and M.~Pizzonia.
\newblock C-planarity of c-connected clustered graphs.
\newblock {\em Journal on Graph Algorithms and Applications}, 12(2):225--262,
  2008.

\bibitem{dahlhaus:latin98}
E.~Dahlhaus.
\newblock A linear time algorithm to recognize clustered planar graphs and its
  parallelization.
\newblock In C.~L. Lucchesi and A.~V. Moura, editors, {\em LATIN~1998}, LNCS vol.
  1380, pp. 239--248. Springer,
  1998.

\bibitem{dibattista/tamassia:icalp90}
G.~{Di Battista} and R.~Tamassia.
\newblock On-line graph algorithms with {SPQR}-trees.
\newblock In M.~Paterson, editor, {\em ICALP~1990}, LNCS vol. 443, pp. 598--611. Springer, 1990.

\bibitem{didimo_etal:2008}
W.~Didimo, F.~Giordano, and G.~Liotta.
\newblock Overlapping cluster planarity.
\newblock {\em Journal on Graph Algorithms and Applications}, 12(3):267--291,
  2008.

\bibitem{feng_etal:esa95}
Q.~Feng, R.~F. Cohen, and P.~Eades.
\newblock Planarity for clustered graphs.
\newblock In P.~Spirakis, editor, {\em ESA~1995}, LNCS vol. 979, pp. 213--226. Springer, 1995.

\bibitem{gutwenger/mutzel:gd2000}
C.~Gutwenger and P.~Mutzel.
\newblock A linear time implementation of {SPQR}-trees.
\newblock In M.~T. Goodrich and S.~G. Kobourov, editors, {\em GD~2002}, LNCS vol. 2528, pp. 77--90. Springer, 2000.

\bibitem{johnson/pollak:87}
D.~S. Johnson and H.~O. Pollak.
\newblock Hypergraph planarity and the complexity of drawing {V}enn diagrams.
\newblock {\em Journal of Graph Theory}, 11(3):309--325, 1987.

\bibitem{kaufmann_etal:gd2008}
M.~Kaufmann, M.~{van Kreveld}, and B.~Speckmann.
\newblock Subdivision drawings of hypergraphs.
\newblock In I.~G. Tollis and M.~Patrignani, editors, {\em GD~2008}, LNCS vol. 5417, pp. 396–--407. Springer, 2009.

\bibitem{patrignani:gdhandbook}
M.~Patrignani.
\newblock Planarity testing and embedding.
\newblock In R.~Tamassia, editor, {\em Handbook of Graph Drawing and
  Visualization}, pp. 1--42. CRC Press, 2014.

\end{thebibliography}
\end{document}